\pdfoutput=1
\PassOptionsToPackage{pdftex,
pdfencoding=auto,
pdfnewwindow=true,
pdfusetitle=true,
bookmarks=true,
bookmarksnumbered=true,
bookmarksopen=false,
pdfpagemode=UseThumbs,
bookmarksopenlevel=1,
pdfpagelabels=true,
breaklinks=true,
colorlinks=true,
unicode=true,
pdfborder={0 0 0},
backref=false
}{hyperref}
\PassOptionsToPackage{usenames,dvipsnames,pdftex}{xcolor}
\documentclass[superscriptaddress,aps,pra,nofootinbib,twocolumn,twoside,reprint,letterpaper,longbibliography,showpacs,10pt,floatfix]{revtex4-2}
\usepackage[utf8]{inputenx}
\usepackage{amsmath,bbm}
\usepackage{amssymb}
\usepackage{graphicx}
\usepackage{color}
\usepackage{amsthm}
\usepackage{dsfont}
\usepackage{float}
\usepackage{cancel}
\usepackage{booktabs}
\usepackage{ragged2e}

\usepackage[alwaysadjust,flushleft,pointlessenum]{paralist}
\setdefaultenum{\bfseries1.}{}{}{}

\usepackage[colorlinks=true, allcolors=blue]{hyperref}
\usepackage{enumitem}   
\usepackage{algorithm} 
\usepackage{algpseudocode} 
\usepackage{MnSymbol}

\newcommand{\cI}{\mathcal{I}}
\newcommand{\cJ}{\mathcal{J}}

\newcommand{\cN}{\mathcal{N}}
\newcommand{\cP}{\mathcal{P}}

\newcommand{\W}{{\rm W}}

\newcommand{\beq}{\begin{equation}}
\newcommand{\eeq}{\end{equation}}
\newcommand{\beqa}{\begin{eqnarray}}
\newcommand{\eeqa}{\end{eqnarray}}
\newcommand{\bra}[1]{\ensuremath{\left\langle#1\right|}}
\newcommand{\ket}[1]{\ensuremath{\left|#1\right\rangle}}

\newcommand{\dd}{\mathrm{d}}

\newcommand{\id}{\mathds{1}}

\renewcommand{\today}{\number\day\space\ifcase\month\or
   January\or February\or March\or April\or May\or June\or
   July\or August\or September\or October\or November\or December\fi
   \space\number\year}

\definecolor{myurlcolor}{rgb}{0,0,0.7}
\definecolor{myrefcolor}{rgb}{0.8,0,0}
\definecolor{purple}{RGB}{128,0,128}
\definecolor{ultramarine}{RGB}{63, 0, 255}
\definecolor{medblue}{RGB}{0, 0, 100}
\definecolor{googleblue}{RGB}{34, 0, 204}
\definecolor{panblue}{RGB}{0,24,150}
\definecolor{carmine}{RGB}{150, 0, 24}
\definecolor{gray}{RGB}{150, 150, 150}

\newcommand{\term}[1]{\textcolor{medblue}{\textbf{#1}}}
\usepackage{amsmath,amsfonts,amssymb,amsthm,mathtools,bm,graphicx,verbatim,mathrsfs,units}
\usepackage{cprotect} %

\newtheorem{thm}{Theorem}
\newtheorem{theorem}[thm]{Theorem}
\newtheorem{prop}[thm]{Proposition}
\newtheorem{proposition}[thm]{Proposition}
\newtheorem{lemma}[thm]{Lemma}
\newtheorem{cor}{Corollary}[thm]

\newtheorem{conj}{Conjecture}
\newtheorem{definition}[thm]{Definition}

\newtheoremstyle{defblock}{0.7\topsep}{0pt}{}{}{}{: }{0pt plus 1pt minus 1pt}{\thmname{\bfseries{#1}}\thmnumber{\bfseries{#2}}\color{medblue}\bfseries\thmnote{#3}}
\theoremstyle{defblock}

\theoremstyle{remark}

\newcommand{\same}[0]{{\textsf{Same}}}
\newcommand{\Bell}[0]{{\textsf{Bell}}}

\newcommand{\tagprop}[1]{\tag{\hyperref[#1]{P\ref{#1}}}}

\usepackage{microtype}
\microtypecontext{spacing=nonfrench}
\microtypesetup{
expansion={true,nocompatibility},
protrusion={true,nocompatibility},
activate={true,nocompatibility},
tracking=true,
kerning=true,
spacing={true}
}
\usepackage[all=normal,floats=tight,mathspacing=tight,wordspacing=tight,paragraphs=normal,tracking=tight,charwidths=tight,mathdisplays=normal,sections=normal,margins=normal]{savetrees}
\everypar=\expandafter{\the\everypar\loosness=-1 }
\usepackage[style=american,autopunct=true]{csquotes}

\usepackage[scr=boondoxupr]{mathalfa} %

\usepackage{hyperref}
\hypersetup{
linkcolor=carmine,
citecolor=googleblue,
urlcolor=panblue,
anchorcolor=OliveGreen}

\AtBeginDocument{%
    \newwrite\bibnotes
    \def\bibnotesext{Notes.bib}
    \immediate\openout\bibnotes=\jobname\bibnotesext
    \immediate\write\bibnotes{@CONTROL{REVTEX42Control}}
    \immediate\write\bibnotes{@CONTROL{%
    apsrev42Control,editor="0",pages="0",title="0",year="1"}}
     \if@filesw
     \immediate\write\@auxout{\string\citation{apsrev42Control}}%
    \fi
}

\begin{document}

\begin{abstract}
We introduce the class of Genuinely Local Operation and Shared Randomness (LOSR) Multipartite Nonlocal correlations, that is, correlations between $N$ parties that cannot be obtained from unlimited shared randomness supplemented by any composition of $(N-1)$-shared causal Generalized-Probabilistic-Theory (GPT) resources. We then show that noisy $N$-partite GHZ quantum states as well as the $3$-partite W quantum state can produce such correlations.
This proves, if the operational predictions of quantum theory are correct, that \emph{Nature's nonlocality must be boundlessly multipartite} in any causal GPT. 
We develop a computational method which certifies that a noisy $N=3$ GHZ quantum state with fidelity $85\%$ satisfies this property, making an experimental demonstration of our results within reach. We motivate our definition and contrast it with preexisting notions of genuine multipartite nonlocality.
This work extends a more compact parallel letter~[Phys. Rev. Lett. 127, 200401
(2021)] on the same subject and provides all the required technical proofs. 
\end{abstract}

\title{Any Physical Theory of Nature Must Be Boundlessly Multipartite Nonlocal}

\date{\today}

\author{Xavier Coiteux-Roy}
\email{xavier.coiteux.roy@usi.ch}
\affiliation{Faculty of Informatics, Università della Svizzera italiana, Lugano, Switzerland.}

\author{Elie Wolfe}
\email{ewolfe@perimeterinstitute.ca}
\affiliation{Perimeter Institute for Theoretical Physics, Waterloo, Ontario, Canada.}

\author{Marc-Olivier Renou}
\email{Marc-Olivier.Renou@icfo.eu}
\affiliation{ICFO-Institut de Ciencies Fotoniques, The Barcelona Institute of Science and Technology, Castelldefels (Barcelona), Spain.}

\maketitle

\twocolumngrid
\section{Introduction}\label{Introduction}

Correlated events are ubiquitous.
A fundamental objective of science is to understand the causal links between these events, behind correlations. Bell's seminal theorem~\cite{bell1964einstein} demonstrated the failure of classical causal theories~\cite{pearl2009causality} to reproduce the predictions of quantum theory. A natural interpretation of Bell's theorem is that the structural links between non-observed underlying variables (also called sources or resources) and observed variables (also called parties) in a network causal model are not sufficient to delimit all possible correlations that might be observed between them: the \emph{physical nature} of the sources is also important~\cite{Wolfe2020quantifyingbell}. Indeed, even in a simple Bell scenario involving one source and several observed variables, a source producing quantum signals allows for ``nonlocal'' correlations that cannot be modelled classically~\cite{CHSHOriginal}. Simply put, the correlations achievable with quantum common causes are richer than those achievable with purely classical sources.

The existence of nonclassical quantum correlations inspired the study of even more general causal theories, capable of explaining quantum correlations and even stronger-than-quantum correlations~\cite{Popescu1994}.
The explanatory power of such exotic theories is so strong that one might wonder if such a theory might describe all the correlations that may be observed in Nature while at the same time never exceeding some measure of complexity. In this article, we focus on the following question: Do there exist causal theories able to model all observable correlations based on \emph{finite-size} nonclassical resources? More precisely, could Nature's correlations be explained by $N$-partite resources, for some finite $N$? 

Unsurprisingly, even some \emph{classical} correlations would be inexplicable in the absence of universal ($N$-way) shared randomness.
In particular, no causal theory restricted to sharing bipartite resources of \emph{any} physical nature could accommodate perfect correlations between three parties~\cite{Hensen2015} (We will show that this no-go result readily generalizes to $N$ parties restricted to $({N{-}1})$-partite resources).
Accordingly, the \enquote{No Shared Randomness} hypothesis is far too strong an assumption in general. Shared randomness is facially an accessible resource: Indeed, $N$ parties can share randomness by simply agreeing on a common stochastic phenomenon to observe, such as the weather. Alternatively, pre-established high-entropy shared randomness can be stored indefinitely in local memories through the use of any number of digital technologies.

As such, in the following we consider shared randomness to always be accessible.
We focus on the (non)simulability of certain $N$-partite correlations in scenarios allowing for the local composition of ${(N{-}1)}$-shared \emph{nonclassical} resources and Local Operations and Shared Randomness (LOSR) between $N$ parties. $N$-partite correlations which \emph{cannot} be simulated in such a scenario are hereafter deemed \emph{genuinely LOSR multipartite nonlocal}.

Some causal theories of correlations generalizing quantum mechanics have already been introduced. 
In particular, \emph{boxworld} is an alternative theory for correlations motivated by nonsignalling boxes~\cite{Janotta2012Boxword}.
Although boxworld produces some correlations which are strictly \emph{beyond} the scope of the predictions of quantum theory, it should also be noted that boxworld cannot reproduce \emph{all} quantum correlations in scenarios with independent sources even when allowing for shared randomness~\cite{Chao2017genuinemultipartite,WeilenmannQuantumBest,Bierhorst2020Tripartite}. 
Here we aim to derive an argument in a theory-agnostic perspective, so that it be compatible with any causal theory. This includes classical; quantum; nonsignalling boxes; and, more generally, any hypothetical causal theory that can be defined in networks.
In the following, we refer to such theories as causal Generalized Probabilistic Theories (GPT) in networks, or more shortly as GPTs. 
It is the role of these theories to define the resources, or states, emitted by each source, as well as the measurements made by each party.
In our theory-agnostic approach, however, we do not refer to, nor rely on, any concrete formalism for GPTs; different ones~\cite{Barrett2007GPT,Short2010couplers} can be used. 
Our unique requirements for the considered theory is to be causal, and to allow for device replication. (These requirements are formalized in Section~\ref{Sec:GenuineMultipNonlocCorr}.)
We call \emph{theory-agnostic} any correlation which can be obtained from such causal theory (equivalent notions are already introduced in \cite{Henson2014,GisinNSI}, see also related work~\cite{Chiribella2011Reconstruction,Chiribella_2014,Bancal2021Networks,Beigi2021,Pironio2021InPreparation}).
The present text extends a more compact parallel letter on the same subject~\cite{PRL} and furnishes all the required technical proofs.

The question of the (non)simulability of certain $N$-partite correlations in setups allowing for the local composition of any ${(N{-}1)}$-shared GPT resources and $N$-shared randomness is intuitively clear. 
Nevertheless, it requires a technical definition of what are \emph{genuinely LOSR $N$-multipartite nonlocal correlations} --- \emph{i.e.}, the correlations which can be obtained through such a process.
In the following Section~\ref{Sec:GenuineMultipNonlocCorr}, we base this definition on a causality principle and device replication, through the inflation paradigm. 
Then, in Section~\ref{sec:GenuineLOSRMultipartitenessGHZW} we prove that the $N$-partite quantum states $\ket{{\rm GHZ}_N}$ can create genuinely $N$-partite nonlocal correlations. 
This proves Theorem~\ref{thm:NatureNotNloc}, the main result of this paper: Nature is not merely ${N}$-partite, for any $N$. Our result is noise tolerant. We also generalize this result to the tripartite state $\ket{{\rm W}}$.
In Section~\ref{sec:NumericalNoiseTolerence}, we provide a linear-programming (LP) method to generate \emph{certificates} of genuine multipartite nonlocality, based on the inflation technique. We illustrate it over the $\ket{{\rm GHZ}_3}$ state, obtaining better noise-robust results accessible to current technologies. Such improvements illustrate the practical importance of this LP method for experimental realizations.
Since there already exist several definitions of the concept of genuinely multipartite nonlocal correlations, we discuss in Section~\ref{Sec:LOCCvsLOSR} the adequacy of ours: an LOSR theory-agnostic framework which optimally accommodates an intuitive concept of genuinely multipartite nonlocal correlations. In particular, we compare our definition to the historically accepted notion due to Svetlichny~\cite{Svetlichny}.

\section{Definition of genuinely LOSR-multipartite-nonlocal correlations}\label{Sec:GenuineMultipNonlocCorr}

In this section, we provide a definition of genuinely LOSR $N$-partite nonlocal correlations. Our approach is closely related to the concept of network nonlocality, which has been a subject of extensive study in the past decade~\cite{TavakoliReview,fritz2012bell,Branciard2010,Renou2019}.
By specializing to the case of $N=3$, the definition herein will precisely formalize the more informal Definition~2 of Ref.~\cite{PRL}. 

As prelude to defining genuine LOSR multipartite nonlocality, we first provide a definition of \emph{${(N{-}1)}$-partite LO theory-agnostic correlations}, that is, of correlations which can be obtained from causal GPT limited entirely to ${(N{-}1)}$-partite resources. 
Then, we extend it to a definition of \emph{${(N{-}1)}$-partite LOSR theory-agnostic correlations}, allowing for $N$-partite shared randomness in addition to ${(N{-}1)}$-partite GPT resources. 
Lastly, we define genuinely LOSR $N$-partite nonlocal correlations as the correlations that are \emph{not} ${(N{-}1)}$-partite LOSR theory-agnostic. 

Recall that standard Bell scenarios involve a single common cause accessible to all parties. In the absence of any particular physical restriction on the nature of that common cause the only \emph{a priori} constraints over such theory-agnostic correlations in a Bell scenario are the No Signalling equalities~\cite{Barrett2007GPT}.
By contrast, in our case theory-agnostic correlations are  restricted by nontrivial inequality constraints in addition to the equality constraints coming from No Signalling. We will show how these inequalities are consequences of the scenario being composed of several independent theory-agnostic sources available only to ${N{-}1}$ parties.

In the following, we base ourselves on a \emph{causality principle} (formalized below, see also its definition in the framework of operational probabilistic theories~\cite{Chiribella2011Reconstruction,Chiribella_2014}) that consists in accepting the causal structure of the scenario.
We also assume that any device distributing a resource, or locally operating on resources, can be replicated in independent copies which can be reordered to form a new setup.
These two ingredients --- causality and device replication --- are all that are needed in order to draw inferences from the nonfanout-inflation technique~\cite{Wolfe2016inflation}; the latter also powers the analytic and computational results in this article.

\subsection{Notations}

\begin{figure*}[htb]\centering
    \includegraphics[width=\linewidth]{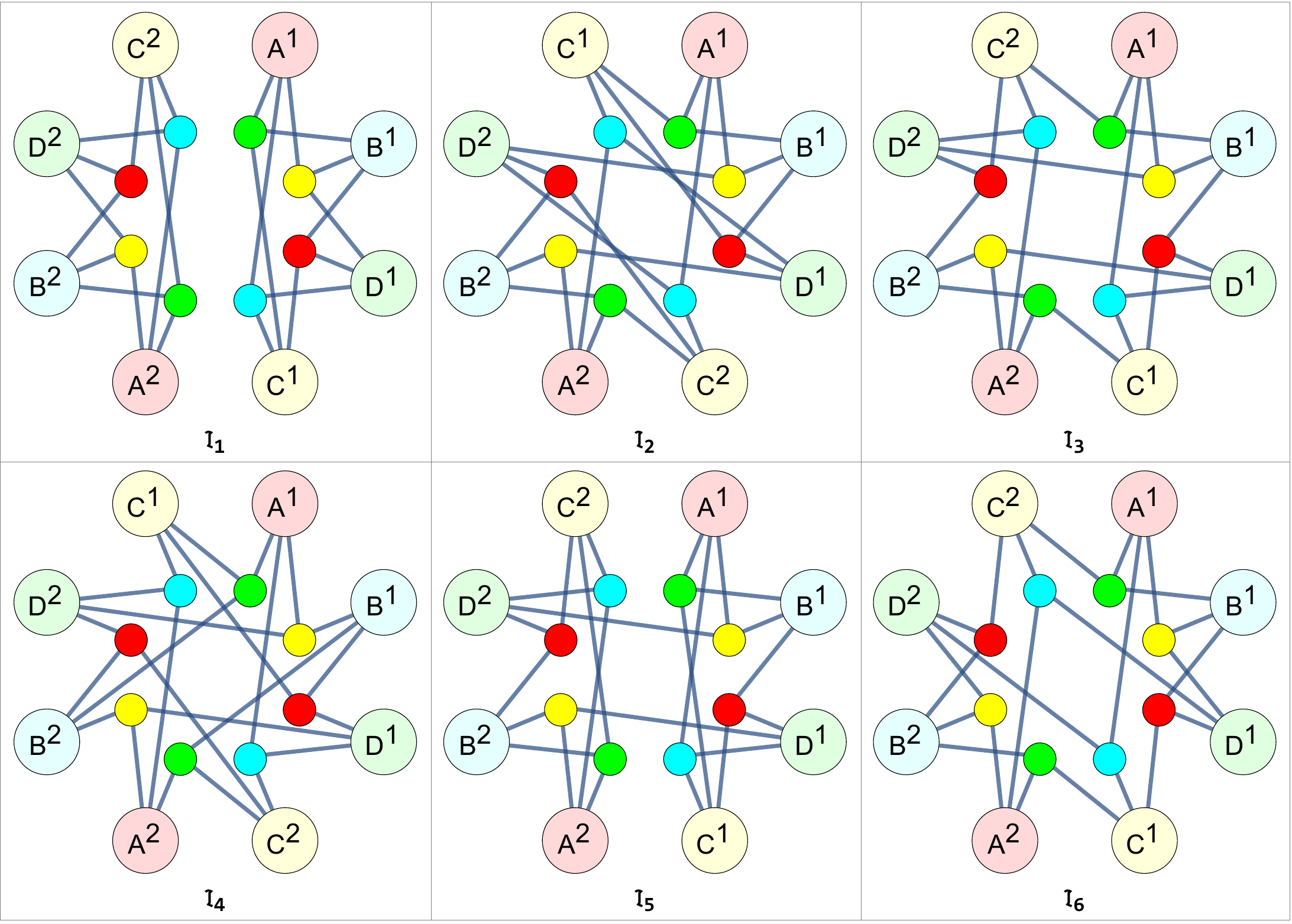}\label{fig:InflationN=4}

  \caption[]{\label{Fig:3WayL2Classes} 
  All nonfanout inflations of order $K=2$ for the tetrahedron network $\cN_4$ (\emph{i.e.}, the network with four 3-way sources). $\cN_4$ is composed of parties $A$ (red), $B$ (blue), $C$ (yellow), $D$ (green) and sources of colours (red, blue, green, yellow) such that each party is connected to each source except for the one of his own colour. 
We represent above the six non-isomorphic inflations $\cI_1, ..., \cI_6$ (which are of respective multiplicity $1,3,3,1,12,12$). 
  Let $P$ be a nonsignalling correlation over $A,B,C,D$.
  For $P$ to be a ${3}$-{LO} theory-agnostic correlation, Definition~\ref{def:LOTheoryAgnosticCorrelation} requires the existence of a correlation $Q^{(1)}, ..., Q^{(6)}$ for each inflated $\cI_1, ..., \cI_6$ such that $(C1)$, $(C2)$ and $(C3)$ are satisfied.
  For example:}\justify
  {\vspace*{-1em}
  \justify\begin{minipage}{\linewidth}
  \begin{compactenum}
  \item[(C1)] implies that ${Q^{(1)}}_{|A^1B^1C^1D^1}=P$ and ${Q^{(2)}}_{|A^1D^2}=P_{|AD}$, but not that ${Q^{(2)}}_{|A^1B^1C^1D^1}=P$. 
  \item[(C2)] implies that every inflation is invariant under the exchange of all copy indices, \emph{e.g.}, ${Q_{|A^1,B^1,C^1,D^1,A^2,B^2,C^2,D^2}}={Q_{|A^2,B^2,C^2,D^2,A^1,B^1,C^1,D^1}}$. 
    \item[(C3)] implies that ${Q^{(3)}}_{|A^1B^1A^2B^2}={Q^{(3)}}_{|A^1B^1}\cdot{Q^{(3)}}_{|A^2B^2}$.
  \end{compactenum}\end{minipage}}
  Note that one can in principle also consider $\cJ={\cI_1,...,\cI_6}$, which is a valid inflation of $\cN_4$ (but of order $K=12$) and which implies the existence of a correlation $Q$ of the parties over $\cJ$ that factorizes as the product ${Q^{(1)}}\cdot...\cdot{Q^{(6)}}$ and satisfies the compatibility conditions imposed by $(C2)$: for instance, it implies $Q^{(3)}_{|A^1B^1C^1A^2B^2C^2}=Q^{(6)}_{|A^1B^1C^1A^2B^2C^2}$.  
\end{figure*}

Let us introduce $\cN_N$, the $N$-partite network scenario in which every $N=\binom{N}{N{-}1}$ subset of ${N{-}1}$ parties is connected to an arbitrary causal GPT resource.
Let $A_1, ..., A_N$ be its parties and $S_1, \dots, S_N$ its sources, such that $A_i$ is connected to every source except for $S_i$ and similarly $S_i$ is connected to every party except for $A_i$. For $N=3$, this corresponds to the triangle network (\emph{without} shared randomness).

Consider now a nonzero integer $K$. We call \emph{$K^{\text{th}}$-order nonfanout inflation of $\cN$} any network $\cI$ composed out of $K$ copies $S_i^1, \dots, S_i^K$ of each source $S_i$ and $K$ copies $A_j^1, \dots, A_j^K$ of each party $A_j$, such that the following inflation-compatibility rules are satisfied:
\begin{itemize}
    \item In $\cI$, any party $A_j^k$ is connected to the same number and same types of sources as in $\cN_N$,
    \item In $\cI$, any source $S_i^{k}$ connects the same number and same types of parties as in $\cN_N$.
\end{itemize}
There exist several nonfanout inflations of a given order. For instance, the case of the triangle ($N=3$) admits two distinct inflations of order $K=2$: one is two copies of the triangle, and the second is a hexagon.
The case of the tetrahedron ($N=4$) has six classes of inflations of order $K=2$ (defined up to graph isomorphism, see Figure~\ref{fig:InflationN=4}). 
Our arguments will often be based on the correlations shared in a sub-network of some large inflation. We call such sub-network an inflation cut. In this paper, most of our figures represent inflation cuts.
In the following, sub-network isomorphisms are of particular interest. A sub-network $G$ of the network $\cN$ or of its inflation $\cI$ consists in a subgroup of parties with all the sources these parties are connected to. 

We say that $(G_1, G_2)$, two sub-networks of ($\cN,\cI$) or of ($\cI,\cI$), are \emph{isomorphic} if they are isomorphic under the dropping of the indices of the parties and sources (because even if two copies of a same party or source have different indices, they are otherwise indistinguishable). A sub-network is defined by an \emph{ordered} list of parties. The ordering means that a nontrivial isomorphism may exist between two sub-networks of $\cI$ even if both sub-networks refer to precisely the same \emph{unordered} set of parties. Hence, an inflated network can be a sub-network of itself in a nontrivial way.
See Figure~\ref{Fig:3WayL2Classes} for an illustration.

In the following, when $R$ denotes a correlation --- that is, a probability distribution of some outputs given some inputs --- in some network $\cJ$, and if $G$ denotes a sub-network of $\cJ$, then $R_{|G}$ represents the marginal distribution of $R$ over the parties in $G$. If $G_1, G_2$ are two non-overlapping (that is, sharing no parties) sub-networks of $\cJ$, we write $G_1\cupdot G_2$ the sub-network of $\cJ$ composed of the parties of $G_1, G_2$ and of the sources they are connected to. We write $R_{|G_1\cupdot G_2}=R_{|G_1}\cdot R_{|G_2}$ to indicate that the marginal distribution can be factorized.

\subsection{Genuinely LO-multipartite-nonlocal correlations}\label{sec:LOMultiNonlocCorr}

We now formalize our causality principle. It first leads to a definition of \textbf{LO} theory-agnostic correlations, which we then extend to \textbf{LOSR} theory-agnostic correlations in Section~\ref{sec:LOSRMultiNonlocCorr}. 

\begin{definition}[${(N{-}1)}$-\textbf{LO} theory-agnostic correlation]\label{def:LOTheoryAgnosticCorrelation}
Consider an $N$-partite nonsignalling correlation $P$.
$P$ is said to be an \term{${(N{-}1)}$-LO theory-agnostic correlation} if, for every nonfanout inflation $\cI$ of $\cN_N$ (of any order), there exists a nonsignalling correlation $Q$ of the parties in $\cI$ such that:
\begin{enumerate}
    \item[(C1)]\label{c1}
    
    For all two $(G_1,G_2)$ sub-networks of $(\cI,\cN_N)$, if the two are isomorphic, then $Q_{|{G_1}}=P_{|{G_2}}$.

    \item[(C2)]\label{c2} For all two $(G_1,G_2)$ sub-networks of $(\cI,\cI)$, if the two are isomorphic, then $Q_{|{G_1}}=Q_{|{G_2}}$.

    \item[(C3)]\label{c3} For all two non-overlapping $(G_1$, $G_2)$ sub-networks of $(\cI,\cI)$, if the two have no sources in common, then $Q_{|{G_1\cupdot G_2}}=Q_{|{G_1}}\cdot Q_{|{G_2}}$.

\end{enumerate}
Note that $(C1)$ is a compatibility condition of $Q$ with $P$, whereas $(C2)$ and$(C3)$ are self-consistency conditions of $Q$ with itself.
\end{definition}

Note that the set of correlations in $\cN_N$ singled out by this definition has already been introduced in other works, under different names. 
In particular, it is the set of the generalized Markov correlations in $\cN_N$, introduced in~\cite{Henson2014}. 
It is also equivalent to the correlations restricted by the No Signalling and Independence principles of~\cite{GisinNSI}. There, $(C1)$ is seen as a consequence of an (extended) No Signalling principle, $(C3)$ is called the Independence principle, and $(C2)$ is implicit. 
In our paper, we view $(C1)$, $(C2)$, and $(C3)$ as consequences of causality. 
More precisely, more than a consequence of causality, these three conditions can actually be seen as the \emph{technical definition} of the intuitive notion of causality, once device replication is allowed.

Before introducing shared randomness, let us first remark that a random bit shared between three parties is not a ${2}$-\textbf{LO} theory-agnostic correlation, as proven in~\cite{Hensen2015}.
The proof can be easily extended to show that for any $N$, $N$-partite shared randomness is not an ${(N-1)}$-\textbf{LO} theory-agnostic correlation: see Figure~\ref{fig:InflSharedRaqndomness}.
As we justified, however, in our introduction, the concept of \textbf{LO} theory-agnostic correlation is not appropriate to discuss the simulability of Nature's correlations, as classical shared randomness is easily accessible.
Hence, we now adapt this definition to take into account a shared source of classical randomness $\lambda$.

\begin{figure}[htb]
    \centering
    \includegraphics[width=\linewidth]{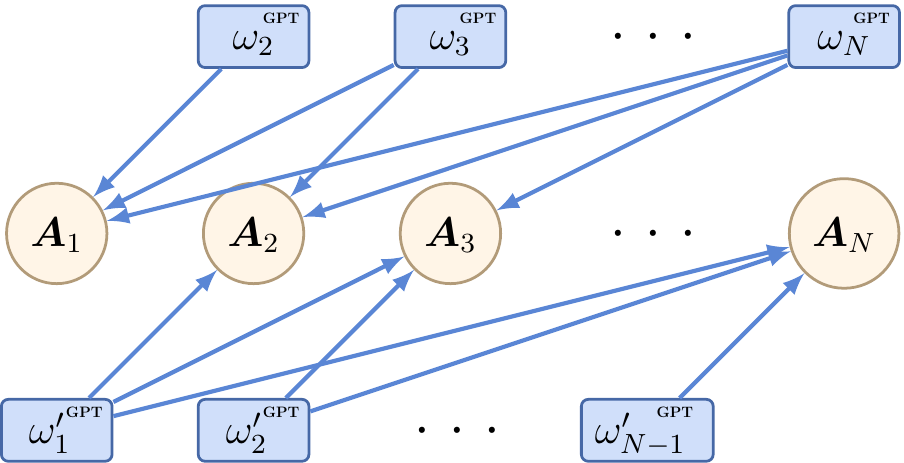}
    \caption{In this inflation $\cI$, each party $A_i$ is connected to the original sources $\omega_j$ for $j>i$ and to the cloned sources $\omega_j'$ for $j<i$.
    Assume by contradiction that there exists an arbitrary setup, with some causal GPT, that allows to simulate a shared random bit in $\cN_N$. In $\cI$, $(C1)$ imposes that two consecutive parties $A_j, A_{j+1}$ share an identical random bit. This implies that any chain of consecutive parties should all together share the same random bit.
    In particular, $(A_1, A_N)$ share the same random bit, which is in contradiction with (C2) as they do not have any sources in common.
    }\label{fig:InflSharedRaqndomness}
\end{figure}

\subsection{Genuinely LOSR-multipartite-nonlocal correlations}\label{sec:LOSRMultiNonlocCorr}

Consider an $N$-partite correlation $P$ which is obtained from a physical process in a scenario involving arbitrary causal GPT resources distributed in $\cN_N$, complemented by shared randomness. 
For any given randomness outcome $\lambda_0$, we obtain a distribution $P_{\lambda_0}$, which is a ${(N{-}1)}$-LO theory-agnostic correlation (note that, \emph{a priori}, it does not have the same marginal as $P$). Writing $\dd\mu(\lambda_0)$ the probability density of a given $\lambda_0$, $P$ can then be written as $P=\int \dd\mu(\lambda) P_{\lambda}$. This discussion motivates the following definition:

\begin{definition}[${(N{-}1)}$-\textbf{LOSR} theory-agnostic correlation]\label{def:LOSRTheoryAgnosticCorrelation}
$P$ is said to be an \term{${(N{-}1)}$-{LOSR} theory-agnostic correlation} if it is a convex mixture of ${(N{-}1)}$-\emph{LO} theory-agnostic correlations. More precisely, the latter implies that there exists a random variable $\lambda$ of density $\dd\mu(\lambda)$ such that $P=\int \dd\mu(\lambda) P_{\lambda}$, and that for every any-order nonfanout inflation $\cI$ of $\cN_N$, there exists nonsignalling correlations $Q_\lambda$ of the parties in $\cI$ such that for all $\lambda$, $Q_\lambda$ satisfies $(C1)$ with respect to $P_\lambda$, as well as $(C2)$ and $(C3)$ with respect to $\cI$.\\
Note that if we introduce $Q\coloneqq\int \dd\mu(\lambda) Q_{\lambda}$, the above conditions imply that $Q$ satisfies $(C1)$ with respect to $P$ and that $Q$ itself satisfies $(C2)$ via linearity of integration. 
\end{definition}

We can now define genuinely \textbf{LOSR}-multipartite-nonlocal correlations.
\begin{definition}[Genuine \textbf{LOSR} multipartite nonlocality]\label{def:GenuineLOSRMultip}
An $N$-partite nonsignalling correlation $P$ is said to be \term{genuinely LOSR multipartite nonlocal} if and only if it is not an ${(N{-}1)}$-LOSR theory-agnostic correlation.
\end{definition}

Note that Definition~\ref{def:LOSRTheoryAgnosticCorrelation}~and~\ref{def:GenuineLOSRMultip} are quite difficult to manipulate, in particular because $(C3)$ is a nonlinear constraint. 
In Section~\ref{sec:NumericalNoiseTolerence}, we propose a relaxation of this set in which we drop this condition. There we show that a weaker --- but more practical --- notion of factorization survives, related to the De Finetti theorem.

\begin{figure}[htb]
    \centering
        \includegraphics{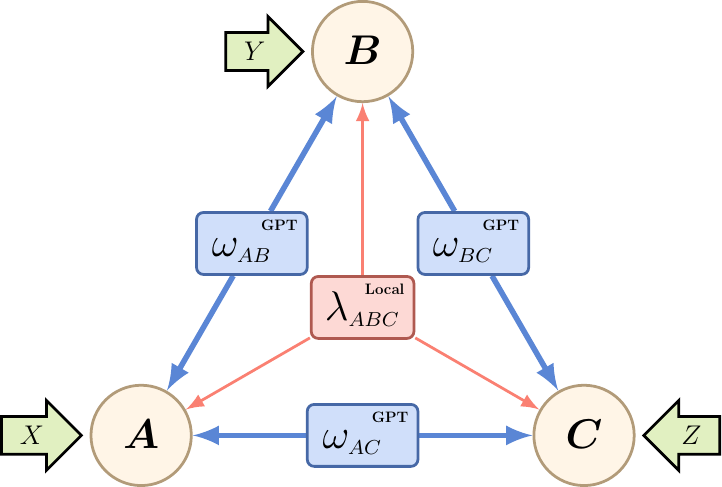}
        
        \caption{A tripartite distribution is \emph{genuinely tripartite nonlocal} according to our Definition~\ref{def:GenuineLOSRMultip} if it is not a ${2-}$LOSR theory-agnostic correlation, that is if it cannot be realized by the above scenario, where the output of each player is determined by local operations (such as joint measurements) on 1)~its input, 2)~the 3-way randomness and 3)~2-way GPT resources.  \label{fig:GPT_triangle}}
\end{figure}

\section{\texorpdfstring{$\ket{{\rm GHZ}_N}$}{ket(GHZ-N)} and \texorpdfstring{$\ket{{\rm W}}$}{ket(W)} create genuinely LOSR-multipartite-nonlocal correlations}\label{sec:GenuineLOSRMultipartitenessGHZW}

In order to explore constraints on ${(N{-}1)}$-\textbf{LOSR} theory-agnostic correlations, we are required to move beyond the case of no-input networks. This is a consequence of the fact that in the presence of shared randomness \emph{any} correlation is compatible with \emph{every} no-input network. Consequently, hereafter we consider exclusively networks with inputs.

In the following Section~\ref{sec:GHZ3}, we show that $\ket{{\rm GHZ}_3}$ can create genuinely $3$-partite nonlocal correlations. 
We also prove a similar result for $\ket{{\rm W}}$ in Section~\ref{Sec:Wstate}. 
Most importantly, we extend this first result in Section~\ref{sec:GHZN} to show that $\ket{{\rm GHZ}_N}$ can create genuinely $N$-partite nonlocal correlations.
This is the main result of this paper, which proves:
\begin{theorem}[Nature is not merely $N$-partite]\label{thm:NatureNotNloc}
Under the hypothesis that quantum mechanics’ predictions for local measurements over $\ket{{\rm GHZ}_N}$ are correct, Nature is not merely $N$-local. More precisely, there exist correlations which cannot be explained by any $N$-partite causal resources and shared randomness.
\end{theorem}
\begin{proof}
This theorem is proven by Proposition~\ref{prop:alt-xavierclaim-general} below. Note that this proof is noise tolerant.
\end{proof}

\subsection{The \texorpdfstring{$\ket{{\rm GHZ}_3}$}{ket(GHZ-3)} quantum state produces genuinely LOSR-tripartite-nonlocal correlations}\label{sec:GHZ3}

\begin{figure*}[htb]\centering
    \includegraphics[width=\linewidth]{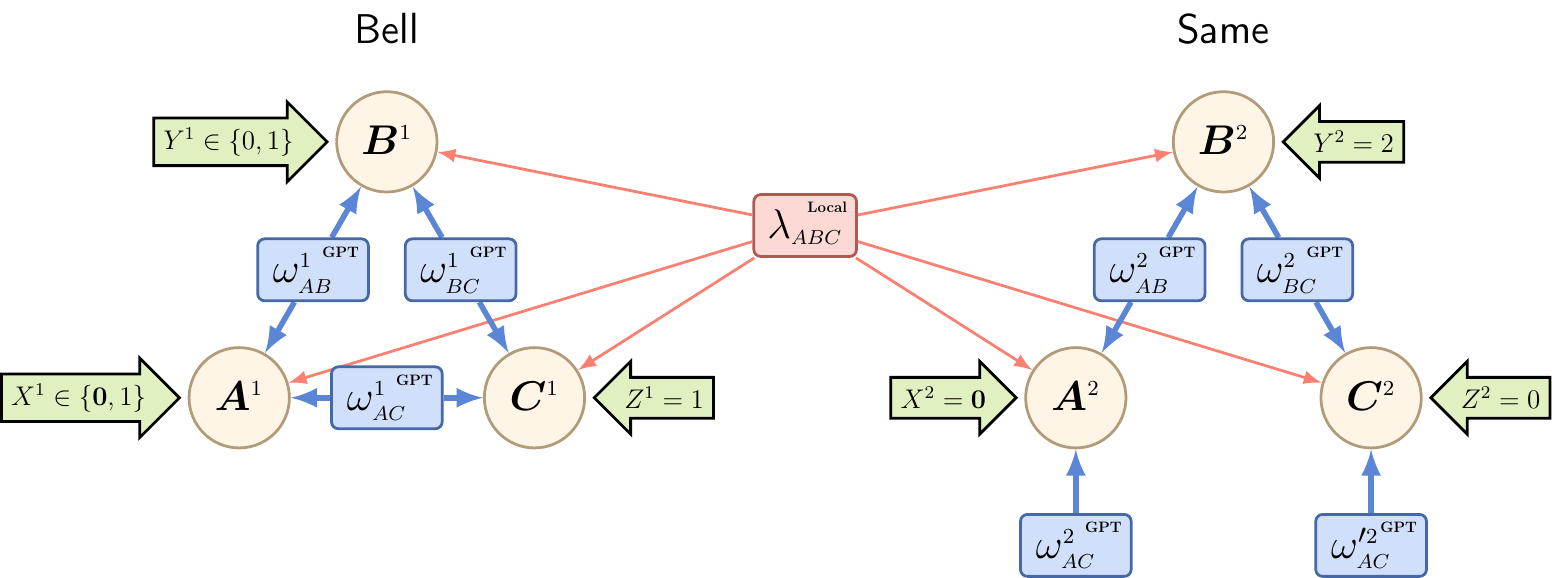}
    \caption{
    The inflation technique consists of duplicating and rearranging players, sources, and input distributions. 
    Here we inflate the (non genuinely tripartite-nonlocal) triangle scenario of Figure~\ref{fig:GPT_triangle} as to have the players play two parallel games ($\Bell$ and $\same$).
    It leads to a contradiction with the statistics of measurements on $\ket{\rm GHZ}$, and therefore to the conclusion that the $\ket{\rm GHZ}$ quantum state is a genuinely tripartite-nonlocal resource.
    The duplicated players are indistinguishable copies of the same abstract process, hence Alice, on input $X{=}\mathbf{0}$, could be playing either game ($A_1$ and $A_2$ must have the same behaviour). The only condition on the random inputs is that they be independent from all of the sources.
    The figure represents a cut of a larger inflation of order 3, consisting of a triangle and a hexagon (three parties of the hexagon are here fully ignored and only the input values relevant for the contradiction are featured).
    \label{fig:MinInflationArgument}}
\end{figure*}

In this section, we prove that the state $\ket{{\rm GHZ}_3}$ can produce genuinely LOSR-tripartite-nonlocal correlations. 
To this end, we first prove the following proposition, which states a constraint for all $2$-LOSR theory-agnostic correlations.
Then, we show that appropriate local measurements of $\ket{{\rm GHZ}_3}$ violate this constraint.
It generalizes Proposition~3 of \cite{PRL} to any value of $\langle C_1\rangle $. This generalization is of particular interest from an experimental perspective.

\begin{prop}[{${\rm GHZ}_3$, technical}\label{prop:ghz3-technical}]
In the absence of any 3-way nonclassical cause,
\begin{align}
 I_\Bell^{C_1{=}1} + \frac{4I_\same}{1+\langle C_1\rangle} \le 6+\frac{4-4\langle C_1\rangle}{1+\langle C_1\rangle}\,.
 \tagprop{prop:ghz3-technical}
 \label{eq:ghz3-technical}
\end{align}
Measurements on the $\ket{\rm GHZ_3}$ quantum state can violate the above by reaching $I_\Bell^{C_1{=}1}+\frac{4I_\same}{1+\langle C_1\rangle}=2\sqrt{2}+8>10$.
\end{prop}

In the above, $I_\Bell^{C_1{=}1}\le4$ and $I_\same\le2$ are respectively defined through the following two tasks:
\begin{enumerate}
    \item[i.] The standard CHSH game between Alice and Bob, with the particularity that it is scored only when Charlie outputs $C{=}1$ (the observables take value in $\{-1,+1\}$):
\begin{align}\label{eq:BellStandard}
I_\Bell^{C_1{=}1}\coloneqq   \langle A_0B_0 &\rangle_{C_1{=}1} + \langle A_0B_1 \rangle_{C_1{=}1} \nonumber\\&+\langle A_1B_0 \rangle_{C_1{=}1} -\langle A_1B_1 \rangle_{C_1{=}1} \,.
\end{align}
    \item[ii.] A game whose goal is for all players to output the same result (\emph{i.e.}, either all $+1$ or all $-1$):
\begin{align}\label{eq:defsame}
I_\same\coloneqq & \langle A_0B_2 \rangle +\langle B_2C_0 \rangle \,.
\end{align}
\end{enumerate}
Note that $A_0\coloneqq A_{X{=}0}$ belongs to both games; Alice is oblivious on that input and thus she cannot adopt a different strategy for the first and second task.
In the following, we first prove the inequality~\ref{eq:ghz3-technical}. Then, we show that the $\ket{\rm GHZ_3}$ state violates it.

\begin{proof}[Proof of Eq.~\eqref{eq:ghz3-technical}]
The proof is based on the inflation argument illustrated in Figure~\ref{fig:MinInflationArgument}. There are four main steps to the proof:

First is the idea behind device-independent randomness certification: \emph{True} randomness is a necessary condition to the violation of Bell inequalities --- if a third party Charlie can guess Alice's input, then Alice and Bob can only win Bell's game with limited success (\textit{i.e.}, $\Bell$ rewards nonclassical resources).
This true randomness is quantified by Theorem~1 of Ref.~\citep[Eq.~(2)]{AugusiakMonogamy}, which states in our case (note that the inequality remains valid when conditioned on $C^1_1=1$ because $C^1$ is space-time separated from $A^1B^1C^2$\/\footnote{Note also that the original theorem in Ref.~\citep[Eq.~(2)]{AugusiakMonogamy} is formulated for (amongst others) the $I_{{\rm BKP}_2}$ Barrett-Kent-Pironio~\cite{BKP} correlations of parameters $M{=}2$ and $d{=}2$, but they are equivalent (up to a relabelling symmetry) to the standard CHSH quantity.})  that

\begin{align}\label{eq:monogamyM-A}\begin{split}
I_{\Bell}^{C^1_1{=}1}\circ \{A^1B^1\} + 2\langle A_0^1C_0^2\rangle_{C^1_1{=}1} \le  4\,.
\end{split}\end{align}
The $\circ$ notation is here used to specify that the $I_{\Bell}^{C^1_1{=}1}$ quantity
is computed over the players Alice-1 and Bob-1, which are the players on the left-hand side of the inflated graph in Figure~\ref{fig:MinInflationArgument} .

Second, we bound $\langle A_0^1C_0^2\rangle_{C^1_1{=}1}$ with $\langle A_0^1C_0^2\rangle$:
For any two events $\{E_1,E_2\}$, the law of total probability implies the bound
\begin{subequations}\begin{align}
&P(E_1,E_2)= P(E_1)-P(E_1,\neg E_2)\label{eq:second_a}\,,\\
\therefore\quad &P(E_1|E_2)= \frac{P(E_1)-P(E_1,\neg E_2)}{P(E_2)}
\ge \frac{P(E_1)-P(\neg E_2)}{P(E_2)}\,.\label{reasoningA8}
\end{align}\end{subequations}
($\neg E_2$ represents the negation of event $E_2$, so $P(E_2)=1-P(\neg E_2)$.)
In our case, we apply the reasoning of Eq.~\eqref{reasoningA8} to the probabilities $P(A^1_0=C^2_0|C_1^1{=}1) = (1+{\langle A^2_0C^2_0 \rangle_{C_1^1{=}1}})/2$ and $P(C^1_1=\pm 1)= (1\pm \langle C^1_1 \rangle)/2$. It leads to the worst-case bound
\begin{align}
\frac{1+\langle A^1_0C^2_0 \rangle_{C_1^1{=}1} }{2} &\ge  \frac{\frac{1+\langle A^1_0C^2_0 \rangle  }{2}-\frac{1-\langle C^1_1\rangle}{2}}{\frac{1+\langle C^1_1\rangle}{2}}\\
\iff  \langle A^1_0C^2_0 \rangle_{C_1^1{=}1} &\ge  \frac{2\langle A^1_0C^2_0 \rangle +2\langle C^1_1\rangle}{1+\langle C^1_1\rangle}-1\,.\label{eq:worst-case-bound}
\end{align}
We use Ineq.~\eqref{eq:worst-case-bound} to rewrite Ineq.~\eqref{eq:monogamyM-A},
\begin{equation}
I_\Bell^{C^1_1{=}1}\circ\{A^1B^1\} +\frac{4\langle A^1_0C^2_0 \rangle +4\langle C^1_1\rangle}{1+\langle C^1_1\rangle}  \le 6\,.\label{equationA44}
\end{equation}

Third, we enter $I_\same$ into the equation:
\noindent We remark that $A^1_0C^2_0 \sim A^2_0C^2_0$ ($\sim$ denotes that the two joint distributions are similarly distributed). This can be seen by observing the inflation in Figure~\ref{fig:MinInflationArgument}: The view of the couple \{Alice-1, Charlie-2\} is exactly the same as the one of the couple \{Alice-2, Charlie-2\}; namely, the joint distributions of all their input resources are identical. One conclusion is hence that 
\begin{align}\label{eq:Isameswitch}
\langle{A^1_0C^2_0}\rangle=\langle{A^2_0C^2_0}\rangle\,.
\end{align}

\noindent We then use an algebraic argument applied to inflation; we find from reformulating Ref.~\citep[App.~A, Eq.~(3)]{NetworkEntanglement2020} that

\begin{align}%
\langle A^2_0 C^2_0\rangle&\geq \langle A^2_0 B^2_2 \rangle + \langle B^2_2 C^2_0 \rangle  - 1
= I_\same\circ\{A^2B^2C^2\} -1 \,.
\label{eq:transitivity-A}
\end{align}

\noindent We now link Ineq.~\eqref{equationA44} and Eq.~\eqref{eq:transitivity-A}, thanks to Ineq.~\eqref{eq:Isameswitch}, and obtain
\begin{align}
    I_\Bell^{C^1_1{=}1}\circ\{A^1B^1\} + \frac{4I_\same\circ\{A^2B^2C^2\}}{1+\langle C^1_1\rangle} \le 6+\frac{4-4\langle C^1_1\rangle}{1+\langle C^1_1\rangle}\,.\label{eqA12}
\end{align}

At last, fourth, we apply the standard lemmas of the inflation technique to recognize that
\begin{subequations}\begin{align}
A^1 B^1C^1X^1Y^1Z^1&\sim ABCXYZ\,,\label{eq:inflationlemma1}\\
A^2 B^2 X^2 Y^2 \sim ABXY&\text{~and~} B^2 C^2 Y^2 Z^2\sim BCYZ\,,
\end{align}\end{subequations}
such that respectively
\begin{subequations}\begin{align}
I_\Bell^{C^1_1{=}1}\circ \{A^1B^1\} = I_\Bell^{C_1{=}1} \text{~and~}\langle C^1_1\rangle=\langle C_1\rangle\,,\\
\text{and~} I_\same\circ\{A^2B^2C^2\} = I_\same\,
.
\end{align}\end{subequations}
As such, Eq.~\eqref{eqA12} --- which applies to the specific inflated-scenario experiment of Figure~\ref{fig:MinInflationArgument} --- is transformed into the general statement of Proposition~\ref{prop:ghz3-technical}.

\end{proof}

\begin{proof}[Proof of violation]
The quantum violation is achieved using $\ket{\rm GHZ}$:
On inputs corresponding to the $\same$ game ($XYZ{=}020$), all players measure in the rectilinear basis. On input $Z{=}1$, Charlie measures his state in the Hadamard basis and obtains marginal $\langle C_1\rangle =0$; when he obtains $C_1{=}1$ (corresponding to $\ket{+}_C$), the state of Alice and Bob is steered towards the maximally entangled state $\ket{\phi^+}_{AB}$ and they can play the $\Bell$ game using the standard optimal strategy for CHSH.\end{proof}

\begin{figure*}[htb!]\centering
    \includegraphics[width=\linewidth]{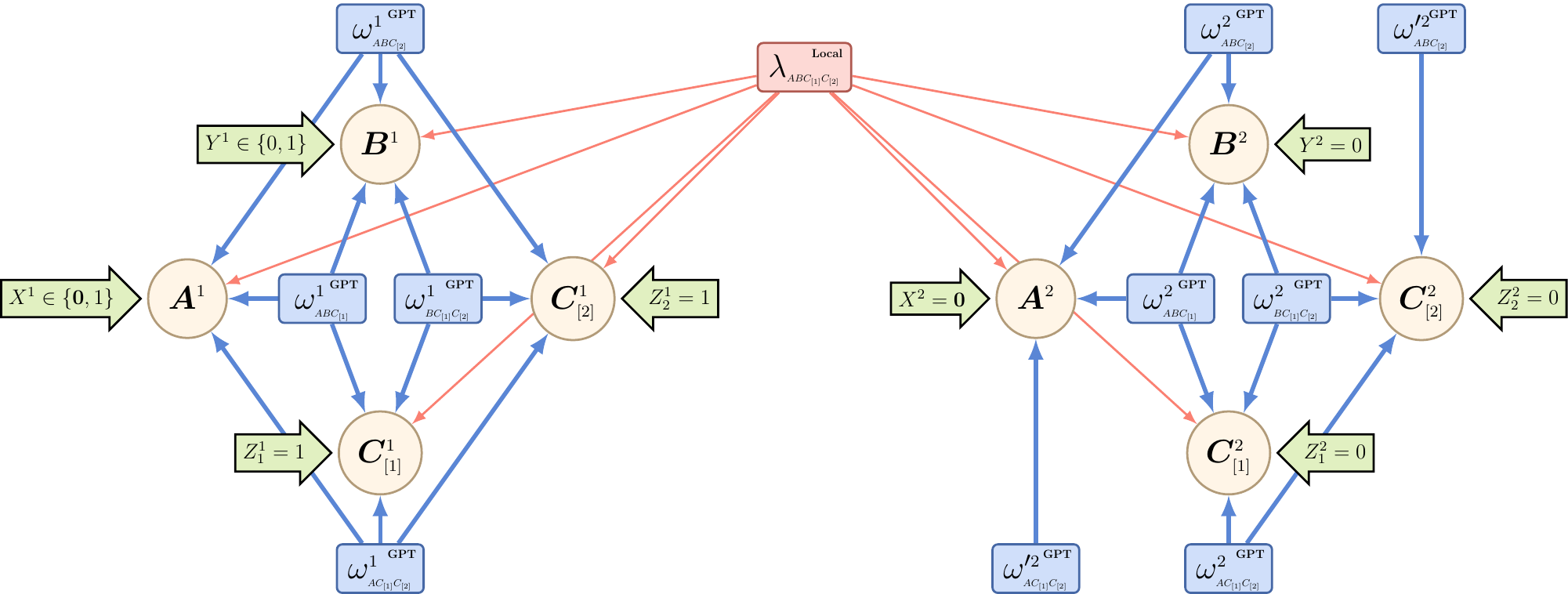}
    \caption{This 4-party nonfanout inflation cut exposes that the quantum state $\ket{\rm GHZ_4}\coloneqq (\ket{0000}+\ket{1111})/\sqrt{2}$ is a genuinely 4-partite nonlocal resource. }
    \label{fig:MinInflationArgument4p}
\end{figure*}

\subsection{The \texorpdfstring{$\ket{{\rm GHZ}_N}$}{ket(GHZ-N)} quantum state produces genuinely LOSR \texorpdfstring{$N$}{N}-multipartite nonlocal correlations}\label{sec:GHZN}

We now generalize Proposition~\ref{prop:ghz3-technical} to all $N$-party scenarios in Eq.~\eqref{eq:mainprop-N} of the following proposition.
The violation of this inequality by the $\ket{{\rm GHZ}_N}$ state (see below) provides the proof of Theorem~\ref{thm:NatureNotNloc} and of Proposition~5 of~\cite{PRL}.
\begin{prop}\label{prop:alt-xavierclaim-general} In the absence of any $N$-way nonclassical common cause,
\begin{align}
    I_{\Bell}^{\tilde{C}_1{=}1} + \frac{4I_{\same_N}}{1+\langle \tilde{C}^1_1\rangle} \le 6+ \frac{4(N-2)-4\langle \tilde{C}^1_1\rangle}{1+\langle \tilde{C}^1_1\rangle} \,,\label{eq:mainprop-N}
\tagprop{prop:alt-xavierclaim-general}
\end{align}
where the game scores $I_{\Bell}^{\tilde{C}_1{=}1}$ and $I_{\same_N}$ are defined below in~\eqref{eq:BellStandard-N} and~\eqref{eq:defsame-N}.
Measurements on the $\ket{{\rm GHZ}_N}$ quantum state can violate the above inequality.
\end{prop}

\begin{proof}[Proof of Eq.~\eqref{eq:mainprop-N}]
This is done by adding, to the two-player argument, extra players whose collective role is similar to Charlie's role in the three-player case. For this reason we name Charlie$_{[i]}$ ($i\in \{1,\dots,N-2\}$) the players that are neither Alice nor Bob. In Figure~\ref{fig:MinInflationArgument4p} we illustrate, for the 4-player case, the inflation scenario on which the following argument is based.

We start by defining the generalization of the two previous, three-player games.
The game that necessitates nonclassical resources to be won maximally is
\begin{align}\label{eq:BellStandard-N}\begin{split}
I_\Bell^{\tilde{C}_1{=}1}\coloneqq  &\langle A_0B_0 \rangle_{\tilde{C}_1{=}1} + \langle A_0B_1 \rangle_{\tilde{C}_1{=}1} \\&\qquad+\langle A_1B_0 \rangle_{\tilde{C}_1{=}1} -\langle A_1B_1 \rangle_{\tilde{C}_1{=}1} \,,\end{split}
\end{align}
where the difference with the three-player game $I_\Bell^{{C}_1{=}1}$ is that $\tilde{C}\coloneqq C_{1[1]} \cdot C_{1[2]} \cdot [\dots] \cdot C_{1[N-2]}$ is defined over the collective of Charlies (\emph{i.e.}, $\tilde{C}_1{=}1$ indicates that all Charlie players had input~1 and an even number of them outputted $-1$);
the game that favours no randomness or genuine tripartite resources (including classical shared randomness) is
\begin{align}\label{eq:defsame-N}\begin{split}
I_{\same_N}\coloneqq \langle A_0B_2\rangle +& \langle B_2C_{0[1]}\rangle+ \langle C_{0[1]}C_{0[2]}\rangle\\&\qquad+[\dots]+ \langle C_{0[N-3]}C_{0[N-2]}\rangle\,.\end{split}
\end{align}

The proof closely follows the one given in Section~\ref{sec:GHZ3} for three parties:

First, as in the three-player case, we use, \emph{mutatis mutandis}, Theorem~1 of Ref.~\citep[Eq.~(11)]{AugusiakMonogamy},
\begin{equation}\label{eq:monogamy-N}
I_{\Bell}^{\tilde{C}^1_1{=}1}\circ \{A^1B^1\} + 2\langle A_0^1C_{0[N-2]}^2\rangle_{\tilde{C}^1_1{=}1} \le  4 \,.
\end{equation}

Second, we bound $\langle A_0^1C_{0[N-2]}^2\rangle_{\tilde{C}^1_1{=}1}$ with $\langle A_0^1C_{0[N-2]}^2\rangle$ and obtain (see Eqs.\eqref{eq:second_a}--\eqref{eq:worst-case-bound})
\begin{equation}
\langle A^1_0C_{0[N-2]}^2 \rangle_{\tilde{C}_1^1{=}1} \ge  \frac{2\langle A^1_0C_{0[N-2]}^2 \rangle +2\langle \tilde{C}^1_1\rangle}{1+\langle \tilde{C}^1_1\rangle}-1\,.
\end{equation}
Therefore,
\begin{equation}
I_{\Bell}^{\tilde{C}^1_1{=}1}\circ \{A^1B^1\} + \frac{4\langle A^1_0C_{0[N-2]}^2 \rangle +4\langle \tilde{C}^1_1\rangle}{1+\langle \tilde{C}^1_1\rangle}\le 6\,.\label{appendixBtherefore}
\end{equation}

Third, we remark from the inflation technique that $A^1_0C_{0[N-2]}^2\sim A^2_0C_{0[N-2]}^2$, so
\begin{equation}
\langle A^1_0C_{0[N-2]}^2 \rangle=\langle A^2_0C_{0[N-2]}^2 \rangle\,.\label{appendixBso}
\end{equation}
We find from applying the recursive algebraic argument of Ref.~\citep[App.~A, Eq.~(27)]{NetworkEntanglement2020} that
\begin{align}\begin{split}
&\langle A^2_0C^2_{0[N-2]}\rangle\geq \langle A^2_0 B^2_2 \rangle + \langle B^2_2 C^2_{0[1]}\rangle+\langle C^2_{0[1]} C^2_{0[2]} \rangle\\&\qquad\qquad\qquad\quad  +[\dots]+\langle C^2_{0[N-3]} C^2_{0[N-2]} \rangle  - N+2\,,\end{split}\\
\shortintertext{or, equivalently,}
&\langle A^2_0C^2_{0[N-2]}\rangle\geq I_\same\circ \{A^2B^2C^2_{[0]}\dots C^2_{[N-2]}\} - N+2\,.\label{eq:transitivity-A-N}
\end{align}

Combining the above (Eqs~\eqref{appendixBtherefore}, \eqref{appendixBso}~and~\eqref{eq:transitivity-A-N}), we get
\begin{align}\label{eq:C7}
\begin{split}
    &\cr I_{\Bell}^{\tilde{C}^1_1{=}1}\circ \{A^1B^1\} + \frac{4I_{\same_N}\circ \{A^2B^2C^2_{[0]}\dots C^2_{[N-2]}\}}{1+\langle \tilde{C}^1_1\rangle}  \\& \cr \le 6+ \frac{4(N-2)-4\langle \tilde{C}^1_1\rangle}{1+\langle \tilde{C}^1_1\rangle} \,.\end{split}
\end{align}

At last, fourth, we recognize using the inflation technique that
\begin{align}
    & I_{\Bell}^{\tilde{C}^1_1{=}1}\circ \{A^1B^1\}=I_{\Bell}^{\tilde{C}_1{=}1}\,,\\
           &\langle \tilde{C}^1_1\rangle=\langle \tilde{C}_1\rangle\,,\\
       &I_{\same_N} \circ \{A^2B^2C^2_{[1]}\dots C^2_{[N-2]}\}=   I_{\same_N} \,;
\end{align}
we conclude the robust statement generalized to $N$ parties:
\begin{equation}\tag{\ref{eq:mainprop-N}}
    I_{\Bell}^{\tilde{C}_1{=}1} + \frac{4I_{\same_N}}{1+\langle \tilde{C}^1_1\rangle} \le 6+ \frac{4(N-2)-4\langle \tilde{C}^1_1\rangle}{1+\langle \tilde{C}^1_1\rangle} \,.
\end{equation}
\end{proof}

\begin{proof}[Proof of violation]
Eq.~\eqref{eq:mainprop-N} admits a violation using measurements on the $N$-partite quantum state $\ket{{\rm GHZ}_N}\coloneqq (\ket{0_1\dots 0_N}+\ket{1_1\dots 1_N})/\sqrt{2}$. The strategy is straightforward --- on inputs corresponding to the $\same$ game, all players measure in the rectilinear basis; on inputs corresponding to the $\Bell$ game, the Charlie players measure in the Hadamard basis (if their product is positive, they have then successfully steered Alice and Bob to the maximally entangled state $\ket{\phi^+}_{AB}$), while Alice and Bob use optimal measurements for the standard Bell game, centred on Alice measuring in the rectilinear basis on input $X=0$. The resulting value for the left-hand side of Eq.~\eqref{eq:mainprop-N} is then ${2\sqrt{2}+4(N{-}1)}$, while the right-hand side is ${4N{-}2}$, hence smaller. This proves that $\ket{{\rm GHZ}_N}$ can produce correlations that are genuinely LOSR $N$-partite nonlocal.
\end{proof}

\subsection{The \texorpdfstring{$\ket{{\rm W}}$}{ket(W)} quantum state produces genuinely LOSR-tripartite-nonlocal correlations}\label{Sec:Wstate}

In this section, we prove that some measurements on $\ket{\rm W}$ lead to correlations that are LOSR genuinely tripartite nonlocal. We use a technique that is similar to the one that we use with $\ket{\rm GHZ}$: a multi-game format analyzed through inflation.
It provides the proof of Proposition~4 of~\cite{PRL}. As opposed to the previous examples, our proof is not noise tolerant.

\begin{prop}[{${\rm W}$}]
In the absence of any 3-way nonclassical cause,
there are quantum measurements on the quantum state $\ket{\rm W}$ that cannot be simulated exactly.
\end{prop}
The global construction rest on a fourth-order inflation (a triangle plus a third-order ring). In the rest of this section, we detail the proof by examining various relevant cuts.

\subsubsection{Preliminaries}
For ease of notation, in this section we take the convention to denote the binary outputs as $\{0,1\}$ rather than as $\{\pm 1\}$.
\paragraph{BKP Inequalities.}
The argument presented in the present section is based on the Barrett--Kent--Pironio correlations of parameter $d{=}2$~\cite{BKP}, which are equivalent to the chained Bell inequalities of Refs.~\cite{pearle1970hidden,braunstein1990wringing} and can be defined as
          \begin{align}\begin{split}
 I_{{\rm BKP}_M}&\coloneqq P(A {=} B|X{=}1,Y{=}M)
 \\&\qquad\qquad+P(A {\neq} B|X{=}M,Y{=}M) \\&\qquad\qquad+  \smashoperator{ \sum_{\substack{i\in \{1,\dots,M-1\}\\j\in \{0,1\}}}} P(A {\neq} B|X{=}i+j,Y{=}i) \,.\end{split}
\end{align}
The inputs have values $X,Y \in \{1,\dots,M\}$. All outputs are binary (\emph{i.e.}, 0 or 1). The algebraic minimum is $I_{{\rm BKP}_M}=0$ but local resources cannot reach less than 1.

The BKP inequalities concern two players, but as for the $\ket{\rm GHZ}$ case, we consider the lifted case where we condition on the outcome $C{=}0$ of a third, space-like separated player (the same local and algebraic minima then apply).
          \begin{equation}\begin{split}
 I_{{\rm BKP}_M}^{C{=}0}&\coloneqq P(A {=} B|X{=}1,Y{=}M,C{=}0)
 \\&\qquad\qquad+P(A {\neq} B|X{=}M,Y{=}M,{C{=}0}) \\&\qquad\qquad+  \smashoperator{ \sum_{\substack{i\in \{1,\dots,M-1\}\\j\in \{0,1\}}}} P(A {\neq} B|X{=}i+j,Y{=}i,{C{=}0}) \,.\end{split}
\end{equation}

Key results concerning BKP correlations are that, in the asymptotic limit $M \to \infty$, the optimal violation allowed by a maximally entangled state $\ket{\phi^+}$ is the algebraic minimum $I_{{\rm BKP}_M}^{C{=}0}=0$~\citep[Eq.~(9)]{BKP}, while for all $\lambda$ for which the output of Alice is completely determined given output $C$, the classical bound of $I_{{\rm BKP}_M}\ge 1$ applies (corollary of Theorem~1 in Ref.~\citep[Eq.~(11)]{AugusiakMonogamy}).

\paragraph{A multi-game which can be won with perfect probability using a \texorpdfstring{$\ket{{\rm W}}$}{ket(W)}-state quantum strategy.}
Similarly to the proof presented for the $\ket{\rm GHZ}$ case, the proof here also follows a multi-game format. Here we define three games which can all be won at the same time using a quantum strategy.
\begin{enumerate}[label=(\roman*)]
    \item\label{BKPgameAB} On inputs $X\in\{1,\dots,M\},Y\in\{1,\dots,M\},Z{=}1$, the players are asked (in the asymptotic limit) to reach $\displaystyle\lim_{M\to\infty} I_{{\rm BKP}_M}^{C{=}0}=0$, while having $P(C{=}0|Z{=}1)=2/3$.
    \item\label{BKPgameBC} A game identical to \ref{BKPgameAB} but with the roles of Alice and Charlie swapped.
    \item\label{One1Two0}  On inputs $XYZ{=}101$ the players must collectively output exactly one ``1'' and two ``0''s (\emph{i.e.}, $ABC\in\{001,010,100\}$).
\end{enumerate}

An important fact is that on input $X{=}1$ (or $Z{=}1$), Alice (or Charlie) does not know which one of the three games she (or he) is playing.

The quantum strategy to win maximally all three games with the $\ket{\rm W}\coloneqq (\ket{001}+\ket{010}+\ket{100})/\sqrt{3}$ state is straightforward. On inputs $X=1$, $Y=0$ and $Z=1$, the players measure in the rectilinear basis and by doing so always win the third game of outputting exactly one $1$. On input $Z{=}1$, Charlie obtains $0$ with probability $2/3$ and in that case the state at Alice--Bob is steered towards a maximally entangle state ${\ket{01}+\ket{10}} / {\sqrt{2}}$. Alice and Bob can then on inputs $X,Y\in\{1,\dots,M\}$ apply the strategy described in Ref.~\cite{BKP} to violate maximally the BKP inequality, achieving asymptotically ${\displaystyle \lim_{M\to \infty} I_{{\rm BKP}_M}^{{A}{B}|{C}{=}0} \to 0}$. The strategy is symmetric in Alice--Charlie and thus the players also violate the BKP inequality when the roles of Alice and Charlie are switched.

\paragraph{``Nonlocal sharing-of-the-one.''}
Our proof relies on a concept that we call ``nonlocal sharing-of-the-one.'' (Note that, without a loss of generality, we consider any private randomness as also part of $\Lambda$.) 
\begin{definition}
The \emph{nonlocal--sharing-of-the-one indicators} are defined for distributions that simulate perfectly (on inputs $XYZ=101$) the classical ${\rm W}$ distribution (\emph{i.e.}, game (iii)); they are
\begin{align}
f_{{A}{B}}^\lambda \coloneqq
\begin{cases} 
      1 & \text{\rm if~}P({A}{=}1|{X}{=}1,\Lambda{=}\lambda)> 0 \\
       & ~~~\mbox{\rm AND~}P({B}{=}1|{Y}{=}0,\Lambda{=}\lambda)> 0\,,\\
    0 & \mbox{\rm otherwise}\,,
   \end{cases}\\
f_{{B}{C}}^\lambda \coloneqq
\begin{cases} 
      1 & \text{\rm if~}P({B}{=}1|{Y}{=}0,\Lambda{=}\lambda)> 0 \\
       &~~~\mbox{\rm AND~}P({C}{=}1|{Z}{=}1,\Lambda{=}\lambda)> 0\,, \\
    0 & \mbox{\rm otherwise}\,,
   \end{cases}\\
   f_{{A}{C}}^\lambda \coloneqq
\begin{cases} 
      1 & \text{\rm if~}P({A}{=}1|{X}{=}1,\Lambda{=}\lambda)> 0 \\
       & ~~~\mbox{\rm AND~}P({C}{=}1|{Z}{=}1,\Lambda{=}\lambda)> 0\,, \\
    0 & \mbox{\rm otherwise}\,. \\
   \end{cases}
\end{align}
\end{definition}

Intuitively $f_{{A}{B}}^\lambda$, for any fixed $\lambda$, is 0 if Alice or Bob (or both) automatically output ``0'' for that $\lambda$ without considering the GPT sources (on inputs $X=1$ and $Y=0$). It equals 1 if both players need the result of manipulations involving GPT resources before ruling out the output ``1.'' The two other indicators have a similar interpretation.

\subsubsection{The proof by contradiction}
We show that the players are not able to complete all three tasks perfectly using shared randomness and merely bipartite resources (while they were able to do so using measurements on the quantum state $\ket{{\rm W}}$). More precisely, in the triangle scenario (see Figure~\ref{fig:GPT_triangle}) succeeding perfectly at all tasks implies two contradictory statements:
\begin{compactenum}
    \item[a.]$\mathop{\mathbb{E}}_\lambda( f_{{A}{B}}^{\lambda}+f_{{B}{C}}^{\lambda})\le 1$. 
    \item[b.] $ \mathop{\mathbb{E}}_\lambda( f_{{A}{B}}^{\lambda}+f_{{B}{C}}^{\lambda})\ge 4/3$.
\end{compactenum}
We prove those two contradictory statements in the subsections below.

\paragraph{\texorpdfstring{Proof that winning perfectly Game~\ref{One1Two0} implies $ \mathop{\mathbb{E}}_\lambda( f_{{A}{B}}^{\lambda}+f_{{B}{C}}^{\lambda})\le 1$.}{Proof part two}}

We look exclusively at the third game and show that simulating perfectly the image of the classical ${\rm W}$ distribution in the LOSR framework with bipartite resources leads to the upper bound $ \mathop{\mathbb{E}}_\lambda( f_{{A}{B}}^{\lambda}+f_{{B}{C}}^{\lambda})\le 1$. In fact, we prove a stronger statement:
\begin{prop}[Monogamy of the nonlocal one]\label{monogamyprop}
In the triangle scenario (see Figure~\ref{fig:GPT_triangle}), when sampling perfectly from the image $\{001,010,100\}$, the following bound must hold.
\begin{equation}
    f_{AB}^\lambda+f_{BC}^\lambda+f_{AC}^\lambda\le 1\,. \label{sharing-of-the-one}
\end{equation}
\label{monogamy1}
\end{prop}
In other words, for each instance $\lambda$ of the shared randomness, at least one player disregards its GPT sources and deterministically outputs $0$.

\begin{proof}
Each term in the sum is by definition either 0 or 1.
We first prove that, for any value $\lambda$ for which ``sometimes Alice outputs 1; and sometimes Charlie outputs 1,'' then ``one of them \emph{will} output 1,'' hence Bob can never output 1 for this $\lambda$. In other terms, we prove:
\begin{equation}\begin{split}
 & \forall \lambda: \left\{ f_{AC}^\lambda= 1  \implies f_{AB}^\lambda=f_{BC}^\lambda=0 \right\}\,.\label{eq-monogamy1_nope}\end{split}
\end{equation}

Our proof use two inflated scenarios. It is a direct corollary of Lemma~\ref{longlinelemma} (below), which uses Lemma~\ref{shortlinelemma} (also below).
Then, the full proposition results from repeating the argument with permuted players (\emph{e.g.}, by replacing $ABC$ and inputs $XYZ{=}101$ with $BAC$ and inputs  $YXZ{=}100$).
\end{proof}

\begin{lemma}\label{shortlinelemma}
In the short-line inflation illustrated in Figure~\ref{line-inflation-1}, $\forall \lambda $ such that $f^\lambda_{AC}{=}1,$
\begin{equation}
  P(A^1C^1\in \{01,10 \}|B^1B^2{=}00,Y^1Y^2{=}00,\Lambda{=}\lambda)=1\,.
\end{equation}

\begin{figure*}[htb]\centering
\includegraphics[width=0.7\linewidth]{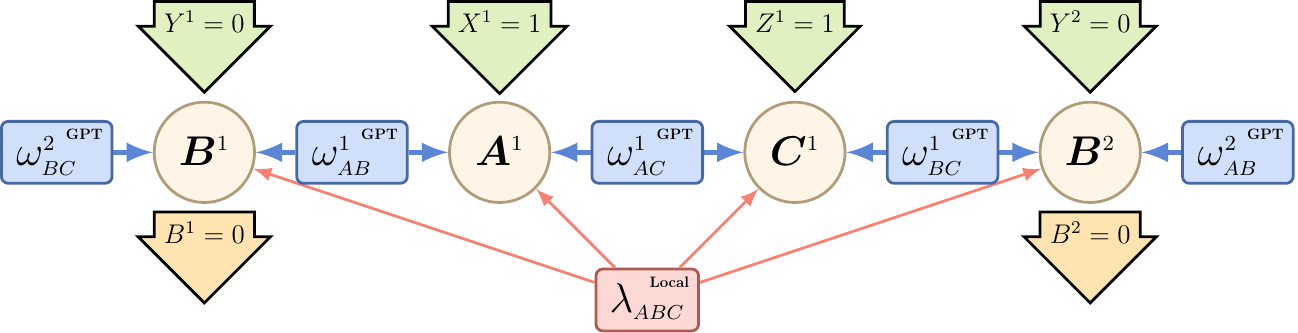}
\caption{Sufficient for Lemma~\ref{shortlinelemma}, this line inflation is in fact a cut of the third-order ring inflation. Given $\lambda$, we condition on the output $B_\lambda^1=0=B_\lambda^2$ (whenever possible) and it imposes $A_\lambda^1C_\lambda^1\in\{01,10\}$.}
\label{line-inflation-1}
\end{figure*}

\begin{proof}

Consider $\lambda$ such that $ f_{AC}^\lambda= 1$. Remark first that we have $P(B{=}1|Y{=}0,\Lambda{=}\lambda)<1$ because if Bob were to always output 1 for that $\lambda$, then to reproduce the image $\{001,010,100\}$ neither Alice nor Charlie could ever output 1 for that $\lambda$.

We use the line inflation depicted in Figure~\ref{line-inflation-1}, and we condition, given $\lambda$, on $B^1=B^2=0$ (it has non-zero weight by the previous remark and because $B^1\sim B \sim B^2$). In what follows, all probabilities are conditioned on the inputs $XYZ{=}101$; we omit them to ease notation.

From inflation, we have \begin{align*}
    &\{A^1B^1\}\sim\{AB\}\,,\\
    &\{B^2C^1\}\sim\{BC\}\,,\\
    &\{A^1C^1\}\sim\{AC\}\,.
    \end{align*}
Because there is exactly one ``1,'' we have
both
\begin{align*}
&P(A{=}1|BC{=}00)=1\,,\\
&P(A{=}1|BC{=}01)=0\,.
\end{align*}
It also holds that 
\begin{align*}\begin{split}
P(A{=}1|B{=}0)=&P(C{=}0|B{=}0) P(A{=}1|BC{=}00)\\+&P(C{=}1|B{=}0) P(A{=}1|BC{=}01)\,.\end{split}\end{align*}

Taken all together, we obtain \begin{align*} P(A^1{=}1|B^1{=}0)+P(C^1{=}1|B^2{=}0)=1\,.\end{align*}
Moreover, we have
\begin{align*}
1=&P(A^1C^1{=}00|B^1B^2{=}00)
+P(A^1C^1{=}01|B^1B^2{=}00)\\
&+P(A^1C^1{=}10|B^1B^2{=}00)
+P(A^1C^1{=}11|B^1B^2{=}00)\,.
\end{align*}
As the two middle terms are\footnote{The second equality holds because $B^1$ and $B^2C^2$ are space-like separated, as well as $B^2$ and $A^1B^1$, and because $P(B^1{=}0|B^2{=}0)>0$ and \emph{vice versa} have non-zero probabilities.}
\begin{align*}
   &P(A^1C^1{=}01|B^1B^2{=}00)+P(A^1C^1{=}10|B^1B^2{=}00)\\ 
    =&P(C^1{=}1|B^1B^2{=}00)+P(A^1{=}1|B^1B^2{=}00)\\
    =&P(C^1{=}1|B^2{=}0)+P(A^1{=}1|B^1{=}0)\\
    =&1\,,
\end{align*}
we have that 
\begin{align*}
P(A^1C^1{=}00|B^1B^2{=}00)=0=P(A^1C^1{=}11|B^1B^2{=}00)\,.
\end{align*} 
In conclusion, for any $\lambda$ for which the conditioning on $B{=}0$ is possible (\emph{e.g.}, all $\lambda$s for which $f_{AC}^\lambda= 1$), we have
\begin{align*}
    &P(A^1C^1{=}00|B^1B^2{=}00,\Lambda=\lambda)=0\,\\
    &P(A^1C^1{=}11|B^1B^2{=}00,\Lambda=\lambda)=0\,.
\end{align*}
Therefore Lemma~\ref{shortlinelemma} holds.

\end{proof}
\end{lemma}

\begin{lemma}\label{longlinelemma}
In the triangle scenario (see Fig~\ref{fig:GPT_triangle}), when sampling perfectly from the image $\{001,010,100\}$,
$\forall \lambda $ such that $f^\lambda_{AC}{=}1,$
\begin{equation}
 P(B{=}1|Y{=}0,\Lambda{=}\lambda)=0\ \,.\label{lemma4eq}
\end{equation}

\begin{figure*}[htb]
\begin{center}
\includegraphics[width=1\linewidth]{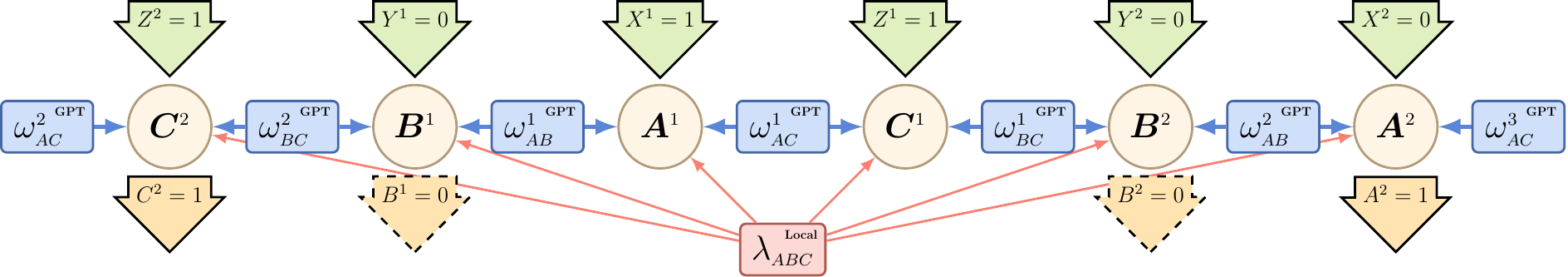}
\caption{
Representing the second step of the proof of Proposition~\ref{monogamy1}, this line inflation is in fact a cut of the third-order ring inflation. Given $\lambda$, we condition on the output $C_\lambda^2=1=A_\lambda^2$, which, to simulate without error a distribution in $\{001,010,100\}$, forces $B_\lambda^1=0=B_\lambda^2$.}
\label{line-inflation-2}
\end{center}
\end{figure*}

\begin{proof}
We consider the slightly extended line inflation of Figure~\ref{line-inflation-2}. We again consider any $\lambda$ for which $P(A{=}1|\lambda)>0$\footnote{The inputs $XYZ=101$ are still omitted.} and $P(C{=}1|\lambda)>0$ (\emph{i.e.}, $f^\lambda_{AC}{=}1$).
Note that $\{C^2A^2|\Lambda{=}\lambda\}\sim\{C^2|\Lambda{=}\lambda\}\{A^2|\Lambda{=}\lambda\}$ ({\it i.e.}, they are independent given $\lambda$).
We use that opportunity to condition on $A^2=1=C^2$.
Because we are reproducing the image of the ${\rm W}$ distribution and because $\{A^2B^2\}\sim\{AB\}$ and $\{B^1C^2\}\sim\{BC\}$, we have
\begin{align*}
    P(B^2{=}0|A^2{=}1)=1\,,\\
    P(B^1{=}0|C^2{=}1)=1\,.
\end{align*}
By Lemma~\ref{shortlinelemma}, we end up with
\begin{align*}
    &P(A^1C^1{\,\in\,} \{01,10 \}|A^2C^2{=}11,\Lambda{=}\lambda)\\
    =&P(A^1C^1{\,\in\,} \{01,10 \}|B^1B^2{=}00,\Lambda{=}\lambda)\\
    =&1
\end{align*}
for all $\lambda$.
Finally, the inflation tells us that
\begin{align*}
    \{A^1C^1|A^2C^2{=}11,\Lambda{=}\lambda\}\sim \{A^1C^1|\Lambda{=}\lambda\}\sim\{AC|\Lambda{=}\lambda\}\,.
\end{align*}
Hence we can remove the conditioning on $A^2C^2{=}11$ and obtain that $1=P(A^1C^1{\,\in\,} \{01,10 \}|\Lambda{=}\lambda)$ (remember that we are assuming $f^\lambda_{AC}{=}1$). Eq.~\eqref{lemma4eq} follows.
\end{proof}
\end{lemma}

\paragraph{\texorpdfstring{Proof that winning perfectly all games imply  $\ \mathop{\mathbb{E}}_\lambda( f_{{A}{B}}^{\lambda}+f_{{B}{C}}^{\lambda})\ge 4/3$.}{Proof part one}}

We look at the first two games --- violating BKP inequalities conditioned on $C{=}0$, or $A{=}0$, respectively\, --- and prove a lower bound that contradicts the upper bound found in \emph{a.}

\begin{prop}
In all nonsignalling LOSR GPT scenarios, if the players succeed perfectly at both Game~\ref{BKPgameAB} and Game~\ref{BKPgameBC}, and also at Game~\ref{One1Two0}, then

\begin{equation}
\ \mathop{\mathbb{E}}_\lambda( f_{{A}{B}}^{\lambda}+f_{{B}{C}}^{\lambda})\ge 4/3\,.
\end{equation}
\end{prop}
\noindent The proof follows from Corollary~\ref{appbcor1}.

We start with a lemma.
\begin{lemma}\label{appBlemma1}
If the players win perfectly Game~\ref{One1Two0}, then $\forall \lambda$ such that $f_{AB}^\lambda=0$, $\forall M$,
\begin{equation}
    I_{{BKP}_M}^{C{=}0}(\lambda)\ge 1\,.
\end{equation}
Note that a similar statement holds when Alice and Charlie are switched.
\begin{proof}
On inputs $XYZ=101$: Given any $\lambda$, if $f_{AB}^\lambda=0$, at least either Alice or Bob outputs deterministically 0.
If it is Alice, then her strategy for Game~\ref{BKPgameAB} is local, and (pre-processing on $C{=}0$ or not) the local bound $I_{{BKP}_M}^{C{=}0}(\lambda)\ge 1$ applies.
If it is Bob that always output 0 for that $\lambda$, then Alice must irremediably output 1 whenever $C=0$. Her strategy for Game~\ref{BKPgameAB} is therefore also local in reference to Bob when conditioned on $C=0$, and the local bound $I_{{BKP}_M}^{C{=}0}(\lambda)\ge 1$ also applies.
\end{proof}
\end{lemma}

\begin{cor}\label{appbcor1}
If the players win perfectly Game~\ref{BKPgameAB} (\emph{i.e.}, $\lim_{M\to \infty} I_{{BKP}_M}^{C{=}0}=0$ and $P(C{=}0|Z{=}1)=2/3$) and Game~\ref{One1Two0},
\begin{equation}
\mathop{\mathbb{E}}_\lambda (f_{{A}{B}}^{\lambda})\ge 2/3\,.
\end{equation}
Note that the equivalent for Game~\ref{BKPgameBC} (replacing Game~\ref{BKPgameAB}) holds, when Alice and Charlie are switched.

\begin{proof}
By contradiction from combining Lemma~\ref{appBlemma1} and $\lim_{M\to \infty}I_{{BKP}_M}^{C{=}0}=0$, we have that $\forall \lambda$ such that $P(C{=}0|Z{=}1,\Lambda{=}\lambda)\neq 0$, $f_{AB}^\lambda=1$. To have $P(C{=}0|Z{=}1)=2/3$, the summed weight of these $\lambda$s must be at least $2/3$. Therefore, $\mathop{\mathbb{E}}_\lambda (f_{{A}{B}}^{\lambda})\ge 2/3$.
\end{proof}
\end{cor}

\section{A computational method to prove the genuineness of LOSR multipartite nonlocality}\label{sec:NumericalNoiseTolerence}
In the previous section, we exhibited several genuinely LOSR-multipartite-nonlocal correlations. 
We proved that the quantum states $\ket{{\rm GHZ}_N}$ and $\ket{\rm W}$ can produce such correlations.
We obtained some (limited) noise-tolerant results for the $\ket{{\rm GHZ}_N}$ state.
In this section, we provide a linear-programming method to obtain certificates of genuine LOSR multipartite nonlocality based on the inflation technique. 
We show that this method improves the noise tolerance that we obtained in the previous section for $\ket{\rm GHZ_3}$. 
The improved noise tolerance makes within experimental reach a demonstration that Nature is not merely bipartite. This method consists in a hierarchy of linear-programming (LP) problems able to characterize a relaxation of the set of LOSR-theory-agnostic correlations.

We first introduce the set of \emph{weakly ${(N{-}1)}$-\textbf{LOSR} theory-agnostic correlations} which we then show can be freely strengthened, and from which we then define an explicit hierarchy.

\begin{definition}[Weakly ${(N{-}1)}$-\textbf{LOSR} theory-agnostic correlation]\label{def:WeakLOSRTheoryAgnosticCorrelation}
Consider an $N$-partite nonsignalling correlation $P$. We say that $P$ is a \term{Weakly ${(N{-}1)}$-LOSR theory-agnostic correlation} if for every any-order nonfanout inflation $\cI$ of $\cN_N$, there exists a nosignalling correlation $Q$ of the party of $\cI$ such that $Q$ satisfies (C1) with respect to $P$, and (C2) with respect to $\cI$.
\end{definition}

Note first that due to the comment at the end of Definition~\ref{def:LOSRTheoryAgnosticCorrelation}, the set of Weakly ${(N{-}1)}$-{LOSR} theory-agnostic correlations is clearly a relaxation of the set of ${(N{-}1)}$-{LOSR} theory-agnostic correlations. 

One can readily anticipate an implementation in terms of linear programs (see Section~\ref{sec:LPHierarchyWeakLOSRMultNonloc}), as the conditions over $Q$ are linear for any fixed inflation $\cI$.
In the following, using De Finetti's theorem, we first show that a looser version of the nonlinear condition $(C3)$ can be derived from the relaxed definition. 

\subsection{A free strengthening of the defining conditions for weakly \texorpdfstring{${(N{-}1)}$}{(N-1)}-LOSR theory-agnostic correlations}

The following Proposition shows that a looser version of the nonlinear condition $(C3)$ is implied by Definition~\ref{def:WeakLOSRTheoryAgnosticCorrelation}. 
\begin{proposition}[Weakly ${(N{-}1)}$-\textbf{LOSR} theory-agnostic correlation]\label{propo:WeakLOSRTheoryAgnosticCorrelation}
Consider an $N$-partite nonsignalling correlation $P$. $P$ is a \emph{Weak ${(N{-}1)}$-LOSR theory-agnostic correlation} if and only if, for every any-order nonfanout inflation $\cI$ of $\cN_N$, there exists a random variable $\lambda$ of density $d\mu(\lambda)$ and nonsignalling correlation $Q_\lambda$ such that, with $Q=\int \dd\mu(\lambda)Q_\lambda$:
\begin{enumerate}
    \item $Q$ satisfies $(C1)$ with respect to $P$: for all two $(G_1,G_2)$ sub-networks of $(\cI,\cN_N)$, if the two are isomorphic, then $\int \dd\mu(\lambda){Q_\lambda}_{|{G_1}}=P_{|{G_2}}$.
    \item $Q$ satisfies $(C2)$ with respect to $\cI$: for all two $(G_1,G_2)$ sub-networks of $(\cI,\cI)$, if the two are isomorphic, then $\int \dd\mu(\lambda){Q_\lambda}_{|{G_1}}=\int \dd\mu(\lambda){Q_\lambda}_{|{G_2}}$.
    \item $Q$ satisfies a loosening of $(C3)$ with respect to $\cI$: For all two non-overlapping $(G_1,G_2)$ sub-networks of $\cI$, if the two have no sources in common, then $\int \dd\mu(\lambda){Q_\lambda}_{|G_1\cupdot G_2}= \int \dd\mu(\lambda){Q_\lambda}_{|G_1}\cdot{Q_\lambda}_{|G_2}$. 
\end{enumerate}
\end{proposition}
Remark that in Definitions~\ref{def:LOSRTheoryAgnosticCorrelation}, the conditions $(C1), (C2), (C3)$ were imposed for every fixed $\lambda$. This proposition replaces all these conditions by weaker versions in which one first integrate over $\lambda$ before imposing the constraint.

\begin{proof}
Consider a $P$ satisfying the proposition's conditions. 
Consider an inflation $\cI$ of $\cN_N$. We need to find a correlation $Q$ of the parties in $\cI$ which decomposes as $Q=\int \dd\mu(\lambda)Q_\lambda$ such that the conditions $1., 2., 3.$ of the proposition are satisfied.
For this, we introduce a larger (infinite) inflation $\cJ=\{\cI_p\}_{p\in\mathrm{N}}$ of both $\cN_N$ and $\cI$, which consists of infinitely many independent copies of $\cI$.
As $P$ satisfies the proposition's conditions, there exists a nonsignalling correlation $R=R_{|\{\cP_p\}_{p\in\mathrm{N}}}$ over all the parties in $\cJ$ such that $R$ satisfies $1.$ and $2.$.

Remark first that as $2.$ is satisfied, for any permutation $\sigma$ of $\mathrm{N}$, the inflation $\cJ^\sigma=\{\cI_{\sigma(p)}\}_{p\in\mathrm{N}}$ which consists in a reordering of the copies of $\cI$ is isomorphic to $\cJ$. Hence, $R$ is invariant under any permutation $\sigma$ of the parties: $R_{|\{\cP_p\}_{p\in\mathrm{N}}}=R_{|\{\cP_{\sigma{(p)}}\}_{p\in\mathrm{N}}}$.
By the De Finetti theorem, this implies that $R$ is a mixture of independent and identically distributed probability distributions over the $\{\cP_p\}$:
\begin{equation}\label{eq:DeFinetti}
    R = \int d\mu(\lambda) (Q_\lambda)^{\otimes\infty}
\end{equation}

We consider the marginal of $R$ over $\cI_0$, a sub-network of $\cJ=\{\cI_p\}_{p\in\mathrm{N}}$, which can be written $Q=R_{|\cI_0}=\int d\mu(\lambda) Q_{\lambda}$. $Q$ can also be seen as a distribution over the parties of $\cI$.
As $R$ satisfies $1.$ and $2.$, $Q$ clearly satisfies $1.$ and $2.$.
Moreover, consider two sub-networks $G_1,G_2$ of $\cI$ which have no sources in common. Then, there exist two ways to see the network $(G_1,G_2)$ as a sub-network of $\cJ$: 
\begin{itemize}
    \item $(G_1,G_2)$ is a sub-network of $\cI_0$, hence of $\cJ$: we call it $G\subset \cI_0\subset \cJ$. Note that $G$ can also be seen as a sub-network of $\cI$.
    \item We can also see $(G_1,G_2)$ as a sub-network of $\cI_1\times \cI_2$ where $G_1\subset \cI_1$ and $G_2\subset \cI_2$. In this case $(G_1,G_2)$ can be seen as a different sub-network of $\cJ$: we call it $G'\subset \cI_1\times \cI_2 \subset \cJ$ .
\end{itemize}
Remark that, as $G_1,G_2$ have no sources in common, the two sub-networks $G$ and $G'$ of $\cJ$ are isomorphic.

Then, we have $ \int \dd\mu(\lambda){Q_\lambda}_{|G_1G_2}= Q_{|G} = R_{|G} = R_{|G'} = \int \dd\mu(\lambda){Q_\lambda}_{|G_1}\cdot{Q_\lambda}_{|G_2}$, where we used the fact that $R$ satisfies $2.$ in the third equality and Eq.~\eqref{eq:DeFinetti} in the fifth equality.
\end{proof}

\subsection{A Linear Programming Hierarchy}\label{sec:LPHierarchyWeakLOSRMultNonloc}

Definition~\ref{def:WeakLOSRTheoryAgnosticCorrelation}
leads to a natural algorithmic way to prove that a nonsignaling correlation $P$ is not a weakly ${(N{-}1)}$-LOSR theory-agnostic correlation. 
As this notion is a relaxation, the success of the algorithm directly implies that $P$ is not a ${(N{-}1)}$-LOSR theory-agnostic correlation, hence that $P$ is genuinely LOSR $N$-partite nonlocal.

Our hierarchy is based on enumerating \emph{all} $K^{\rm th}$-order inflations $\cI_1, ..., \cI_{p_K}$ of $\cN_N$, requiring a $Q$ satisfying Definition~\ref{def:WeakLOSRTheoryAgnosticCorrelation} for each of them, but also imposing cross-inflation compatibility constraints, which would normally only show up at a higher-order inflation. Nevertheless, adding these extra constraints does not require any increase in the number of variables in the linear program, and hence it would be wasteful in practice not to include them.
\begin{definition}[The $K^{\rm th}$-order inflation test for evaluating if $P$ might be a weakly ${(N{-}1)}$-LOSR theory-agnostic correlation]
Consider an $N$-partite nonsignalling correlation $P$. Then, a necessary condition for $P$ to be a weakly ${(N{-}1)}$-LOSR theory-agnostic correlation is that for every nonfanout inflation $\cI$ of $\cN_N$ (up to order $K$), there exists a nonsignalling correlation $Q^{(\cI)}$ of $N\times K$ parties such that:
\begin{enumerate}
    \item[(C1)]
    For all two $(G_1,G_2)$ sub-networks of $(\cI,\cN_N)$, if the two are isomorphic, then $Q^{(\cI)}_{|{G_1}}=P_{|{G_2}}$.
    \item[(C2+)] For all two $(G_1,G_2)$ sub-networks of a pair of $K^{\rm th}$-order inflations $(\cI_1,\cI_2)$, including but not limited to the special case $\cI_1=\cI_2$, if $G_1$ and $G_2$ are isomorphic, then $Q^{(\cI_1)}_{|{G_1}}=Q^{(\cI_2)}_{|{G_2}}$.
\end{enumerate}
\end{definition}

This algorithm is a direct adaptation of the algorithms presented in the original papers on the inflation technique~\cite{Wolfe2016inflation,Navascues2017completion}, hence we only sketch it.

\newcommand{\blockindent}[1]{{\addtolength{\leftskip}{\algorithmicindent}\relax  {\noindent #1 \par}}}

\begin{algorithm}[H]
	\caption{The $K^{\rm th}$-order inflation test for evaluating if $P$ might be a weakly ${(N{-}1)}$-LOSR theory-agnostic correlation} \label{algo:WeakLOSTthAgnCorr}
	\begin{algorithmic}[1]

	    \State \textbf{INPUT:} \emph{An $N$-partite nonsignalling correlation $P$ and an integer $K$ specifying the hierarchy order}
			\State Enumerate all $K^{\rm th}$-order inflations $\cI_1, ..., \cI_{p_K}$ of $\cN_N$.
		
			\For{$i=1,\ldots,p_K$}

								\vspace{2pt}\State
		
				    \vspace{-\baselineskip}
				\blockindent{Find $A_1^{(i)}$, $B_1^{(i)}$ such that the linear-compatibility conditions $(C1)$ between the unknown correlation $Q^{(\cI_i)}$ and the distribution $P$ can be written as ${A_1^{(i)}\cdot Q^{(\cI_i)} =B_1^{(i)}}$.}

				\For {$j=i,\ldots,p_K$}
				
								\vspace{2pt}\State
				
				    \vspace{-\baselineskip} \blockindent{\blockindent{Find $A_2^{(i)}$, $A_3^{(j)}$ such that for every pair of isomorphic subgraphs of $\cI_i$ and $\cI_j$ the linear-compatibility conditions $(C2+)$ between the unknown marginal correlations $Q^{(\cI_1)}_{|{G_1}}$ and $Q^{(\cI_2)}_{|{G_2}}$ can be captured by constraints ${A_2^{(i)}\cdot Q^{(\cI_i)} =A_3^{(j)}\cdot Q^{(\cI_j)}}$.}}
				\EndFor
		\EndFor

			\State Solve the Linear Program (LP) regarding the existence of vectors $0\leq Q^{(\cI_1)},...,Q^{(\cI_{p_K})}\leq 1$ such that \begin{compactitem}[\hspace{15pt}\raisebox{0.25ex}{\tiny $\bullet$}]
			    \item for all $i\in\{1,\ldots,p_K\}$, each $Q^{(\cI_i)}$ is a correlation;
			    \item and for all $i\in\{1,\ldots,p_K\}$, the correlation satisfies ${A_1^{(i)}\cdot Q^{(\cI_i)} =B_1^{(i)}}$;
			    \item and for all $i,j\in\{1,\ldots,p_K\}$ such that $i\leq j$, the pairs of correlations satisfy ${A_2^{(i)}\cdot Q^{(\cI_i)} =A_3^{(j)}\cdot Q^{(\cI_j)}}$.
			\end{compactitem}

			\If  {LP is infeasible (\emph{i.e.}, the constraints cannot be simultaneously satisfied)}

								\vspace{0pt}\State
				
				    \vspace{-\baselineskip}
				    \blockindent{Output ``$P$ is not a Weakly ${(N{-}1)}$-LOSR theory-agnostic Correlation.''}
			\EndIf
	\end{algorithmic} 
\end{algorithm}

Details of the practical technicalities involved with formulating the appropriate $A$~matrices and $B$~vectors can be found in Ref.~\cite{Wolfe2016inflation}. Infeasibility of the LP indicates that $P$ is not a weak ${(N{-}1)}$-LOSR theory-agnostic correlation, and hence that $P$ is genuinely LOSR $N$-partite nonlocal.

Note that in finite time one can only ever test $K^{\rm th}$-order inflations up to some finite $K$. Hence, in finite time, the algorithm cannot prove that $P$ \emph{is} a weakly ${(N{-}1)}$-LOSR theory-agnostic correlation (but it can prove it is \emph{not}). 
In other words, this algorithm is only useful in proving the genuine LOSR $N$-partite nonlocality of a distribution; happily, this is precisely our goal.

\subsection{A better noise tolerance for \texorpdfstring{$\ket{{\rm GHZ}_3}$}{ket(GHZ3)}}

Our proposition~\ref{prop:ghz3-technical} in Section~\ref{sec:GHZ3}, which generalizes Proposition~3 of~\cite{PRL}, proves that the state $\ket{{\rm GHZ}_3}$ can produce genuinely LOSR-tripartite-nonlocal correlations. 
In this section, we focus on the noise tolerance of this claim: With a noisy source of $\ket{{\rm GHZ}_3}$ states, can one still observe genuinely LOSR-tripartite-nonlocal correlations? 
This question is of particular interest for experimental concerns.

For simplicity, we focus on the case of white noise, for a state measured with optimal measurements operators (the following can be generalized to other noise models).
We consider a source emitting a mixture of $\ket{{\rm GHZ}_3}$ with the maximally mixed state,
\begin{equation}
\rho_p = p\ket{{\rm GHZ}_3}\bra{{\rm GHZ}_3} + (1-p) \id / 8\,,
\end{equation}
and look for conditions on $p$ ensuring that $\rho_p$ can demonstrate genuinely LOSR-tripartite-nonlocal correlations.
The fidelity of $\rho_p$ with $\ket{{\rm GHZ}_3}$ is $f=\bra{GHZ_3}\rho_p\ket{{\rm GHZ}_3}=(1+7p)/8$, \emph{i.e.}, $p=(8f-1)/7$.

Remark first that Proposition~\ref{prop:ghz3-technical} already allows to find a first noise-tolerant bound: with $\rho_p$, performing the same measurements as in the ideal protocol, $I_{Bell}^{C_1=1}[\rho_p]=p\cdot 2\sqrt{2}$, $I_{Same}[\rho_p]=p\cdot 2$ and $\langle C_1\rangle=0$, hence~\eqref{eq:ghz3-technical} is violated as long as $p\cdot(2\sqrt{2}+8)>10$. Hence, we obtain a first proof of genuine LOSR tripartite nonlocality for $p \gtrsim92\%$, corresponding to a fidelity $f \gtrsim 93\%$. 
This bound is experimentally challenging. 
For instance, recent experimental work could prove a violation of Mermin and Svetlichny inequalities with a three-photon $\ket{{\rm GHZ}_3}$ state of fidelity $\sim 86\%$~\cite{Hamel2014GHZ3MerminSvetlichny}.

To improve the noise tolerance of our result, we implemented Algorithm~\ref{algo:WeakLOSTthAgnCorr} using Mathematica and evaluated it for the inflation test of order ${K{=}2}$.
Considering again the correlation obtained by measuring $\rho_p$ with the same measurements as in the ideal protocol of Proposition~\ref{prop:ghz3-technical}, we obtained a certificate of LOSR tripartite nonlocal genuiness for all state with $p > 2\sqrt(2)-2 \approx 83\%$, \emph{i.e.}, of fidelity $f \gtrsim 85\%$.

This improvement shows the importance of the computation approach for experimental proofs of the claim that Nature is not merely $N$-partite for low $N$ (see~\cite{Hamel2014GHZ3MerminSvetlichny,GHZExperiment6Photons,Chao2019,Proietti2019} for experimental capabilities up to $N=6$). 
Considering higher order inflation tests may result in better noise-tolerant results. We also emphasize that Algorithm~\ref{algo:WeakLOSTthAgnCorr} can also be applied to the $\ket{{\rm GHZ}_N}$ and $\ket{{\rm W}}$ cases, possibly with alternative quantum measurements.

\section{On LOCC vs LOSR and everything in between}\label{Sec:LOCCvsLOSR}

\subsection{A generalization of \texorpdfstring{$(N{-}1)$}{N-1}-theory-agnostic correlations to \texorpdfstring{$k$}{k}-theory-agnostic correlations}

In Section~\ref{Sec:GenuineMultipNonlocCorr}, we introduced Definition~\ref{def:LOSRTheoryAgnosticCorrelation}, namely ${(N{-}1)}$-LOSR theory-agnostic correlations, for describing the set of $N$-partite correlations that can be obtained by local composition with any GPT ${(N{-}1)}$-partite resources as well as as $N$-partite shared randomness.

Here we firstly note, in an informal way, that this definition can easily be altered to characterize instead the $N$-partite theory-agnostic correlations that can be obtained by allowing for shared randomness alongside GPT \emph{$k$-partite resources}, for some $k<N$.
To do this, one needs simply consider the $N$-party network scenario in which every subset of $k$ parties is connected to an arbitrary causal GPT resource. One can then proceed, as before, by considering $K^{th}$-order nonfanout inflation of this network scenario.

\subsection{Several definitions of genuine multipartite nonlocality}

We defined as genuinely LOSR $N$-multipartite nonlocal the correlations which are not ${(N{-}1)}$-theory-agnostic. One can, however, consider a variety of related definitions of genuinely multipartite-nonlocal distributions, and of associated $k-$theory-agnostic correlations. Here we enumerate some of them, and discuss on how they interrelate to one another.

In the following, we consider various physical causal theories for correlations such as the classical, quantum and boxworld theories, or the signalling-boxes theory, which allows for signalling distributions. We consider the set of all $k-$partite \emph{$\mathcal{R}$-like boxes}, that is, all correlations which can be obtained in a $k-$party scenario in a theory $\mathcal{R}$. We call ${\boldsymbol{P}}_{\mathcal{R}}^{\leq k}$ this set of correlations. Given $\mathcal{R}$, one can define a notion of LOCC nonclassicality as follows:

\begin{definition}[Genuine LOCC-$\mathcal{R}$ tripartite nonlocality]
A tripartite nonsignalling correlation $P$ is said to be \term{LOCC-$\mathcal{R}$ tripartite producible} if $P$ can be decomposed into a convex mixture of products of onepartite and bipartite $\mathcal{R}$-like boxes, \emph{i.e.}, correlations in ${\boldsymbol{P}}_{\mathcal{R}}^{\leq 2}$.\\
A distribution which is not LOCC-$\mathcal{R}$ bipartite producible is said to be \term{genuinely LOCC-$\mathcal{R}$ tripartite nonlocal}.
\end{definition}

In this definition, the term LOCC indicates that one is allowed to perform local operation over the \emph{classical} inputs and output of $\mathcal{R}$-like boxes, such as post-processing or wiring. Local operations over the \emph{physical states} in theory $\mathcal{R}$ are not allowed. In particular, when $\mathcal{R}\to {\mathcal{Q}}$ is quantum theory, entanglement swapping is not an allowed operation as it cannot be performed via local operation on some quantum measurement classical outputs.

This definition specializes to the standard Svetlichny notion of multipartite nonlocality (Definition~2 of \cite{PRL}) upon taking ${\boldsymbol{P}}_{\mathcal{R}}^{\leq 2}$ to be the set of all 
(onepartite and bipartite) correlations, including signalling correlations, \emph{i.e.} ,
$\mathcal{R}\to {\mathcal{S}\mathscr{ig}}$. One can also take 
$\mathcal{R}\to {\mathcal{NS}}$
to be the set of all (onepartite and bipartite) nonsignalling boxes, such as in Refs.~\cite{Curchod2015MultipartiteNonlocality,Baccari2018NonlocalityDepth}. Additional significant choices for $\mathcal{R}$ include the $\mathcal{R}\to {\mathcal{Q}}$
--- obtaining ${\boldsymbol{P}}_{\mathcal{Q}}^{\leq 2}$, the set of all quantum (onepartite and bipartite) correlations --- as well as 
$\mathcal{R}\to \mathcal{TOBL}$
--- obtaining ${\boldsymbol{P}}_{\mathcal{TOBL}}^{\leq 2}$, the set of all (onelocal and bilocal) time-ordered (TOBL) correlations, see Refs.~\cite{Gallego2012TOBL,bancal2013definitions}.

Ref.~\cite{Curchod2015MultipartiteNonlocality} provides a quantitative generalization of these LOCC-centric definitions of multipartite nonlocality to more than three parties: 
\begin{definition}[LOCC-$\mathcal{R}$ minimal group size]
An $N$-partite correlation $P$ is said to be \term{LOCC-$\mathcal{R}$ $k$-partite producible} if $P$ can be decomposed into a convex mixture of products of \emph{at most} $k$-partite $\mathcal{R}$-like boxes, \emph{i.e.}, correlations in ${\boldsymbol{P}}_{\mathcal{R}}^{\leq k}$.\\
 The \term{LOCC-$\mathcal{R}$ minimal group size} of a correlation $P$ is the smallest $k$ such that $P$ is LOCC-$\mathcal{R}$ $k$-partite producible.
\end{definition}

A plurality of definitions can similarly be encompassed within a spectrum of notions of LOSR multipartite nonlocality. 

It is important to keep in mind that in the LOCC-$\mathcal{R}$ sub-definitions, one is limited to the type \enquote{$\mathcal{R}$} \emph{boxes}, producing the correlations ${\boldsymbol{P}}_{\mathcal{R}}^{\leq k}$, to find a convex decomposition of $P$. 
By contrast, in the following LOSR-$\mathcal{S}$ sub-definitions, one can use some family of $k$-way \emph{sources} $\omega^{\leq k}_{\mathcal{S}}$ that are comprising the elementary constituents of a physical network. 

For any class of type-$\mathcal{S}$ sources which can serve as nonclassical resources in a physical network we have:
\begin{definition}[Genuine LOSR-$\mathcal{S}$ tripartite nonlocality]
A tripartite nonsignalling correlation $P$ is said to be \term{LOSR-$\mathcal{S}$ tripartite producible} if $P$ can be obtained by local operations over any 2-way $\mathcal{S}$-type resources $\omega^{\leq k}_{\mathcal{S}}$ along with 3-way shared randomness between all parties.\\
A distribution which is not LOSR-$\mathcal{S}$ tripartite producible is said to be \term{genuinely LOSR-$\mathcal{S}$ tripartite nonlocal}.
\end{definition}
We define bipartite GPT states as the states which allow to recover our definition of genuinely LOSR-multipartite-nonlocal correlations for $N=2$, that is, the states ${\boldsymbol{\omega}}_{\mathcal{GPT}}^{\leq 2}$ recovers Definition~\ref{def:GenuineLOSRMultip}. Alternatively, one could take the 2-way resources to be the quantum states ${\boldsymbol{\psi}}_{\mathcal{Q}}^{\leq 2}$. This \emph{quantum} causal notion of LOSR multipartite nonlocality is equivalent to the definition of 3-way nonlocality given in Ref.~\cite{LOCCInappropriate}. One could also imagine explicit nonclassical theories distinct from quantum theory, such as explicit variants of the \emph{boxworld} GPT~\cite{Janotta2012Boxword}. 

The multipartite generalization is as follows:
\begin{definition}[LOSR-$\mathcal{S}$ minimal group size]
An $N$-partite correlation $P$ is said to be \term{LOSR-$\mathcal{S}$ $k$-partite producible} if $P$ can be obtained via local operations acting on some network consisting of various $k$-way $\mathcal{S}$-type sources along with $N$-way classical randomness shared between all parties.\\
 The \term{LOSR-$\mathcal{S}$ minimal group size} of a correlation $P$ is the smallest $k$ such that $P$ is LOSR-$\mathcal{S}$ $k$-partite producible.
\end{definition}

\subsection{Networks with sources distributing nonlocal boxes instead of entangled states}

To assess LOSR multipartite nonlocality, we are imagining $N$-partite networks wherein every size $k$ subset of parties shares a nonclassical source. Accordingly, every individual party is connected to $\binom{N{-}1}{k-1}$ distinct sources. (For the triangle we have $N{=}3$ and $k{=}2$, and every party is connected to two sources.)

This manuscript is concerned with GPT sources ${\boldsymbol{\omega}}_{\mathcal{GPT}}$, \emph{i.e.}, sources which distribute GPT entanglement. We have also alluded to sources which distribute \emph{quantum} entanglement ${\boldsymbol{\psi}}_{\mathcal{Q}}$. In all such cases, we are considering the sources themselves to be described as multipartite entangled \emph{states}. One can, however, imagine a network in which the $k$-way sources connecting noncommunicating parties are taken to be \emph{nonlocal boxes} instead of entangled states. 

That is to say, in addition to possibly considering ${\boldsymbol{\omega}}_{\mathcal{GPT}}^{\leq k}$ and ${\boldsymbol{\psi}}_{\mathcal{Q}}^{\leq k}$ we can imagine to consider ${\mathcal{S}\to {\boldsymbol{P}}_{\mathcal{NS}}^{\leq k}}$ (resp. ${\mathcal{S}\to {\boldsymbol{P}}_{\mathcal{Q}}^{\leq k}}$), that is to consider correlations in ${\boldsymbol{P}}_{\mathcal{NS}}^{\leq k}$ (resp. ${\boldsymbol{P}}_{\mathcal{Q}}^{\leq k}$) as our sources. 

When the sources in a network are themselves multipartite entangled states, then the local operations performed by a single party (say, Alice) are described by \term{entangled measurements} applied to Alice's subspaces within her $\binom{N{-}1}{k-1}$ connected sources. By contrast, when the sources in a network are themselves nonlocal boxes, then the local operations performed by Alice are described by \term{wirings} that she applies to her portions of the $\binom{N{-}1}{k-1}$ nonlocal boxes that she is connected to. 

Entangled measurements are more general than wirings. Accordingly, the set of correlations realizable using a network of $k$-way quantum sources (${\boldsymbol{\psi}}_{\mathcal{Q}}^{\leq k}$) includes the set of correlations realizable using a network of $k$-way quantum-correlation boxes (${\mathcal{S} \to {\boldsymbol{P}}_{\mathcal{Q}}^{\leq k}}$).

It is worth emphasizing that the inclusion is \emph{strict} however.

\begin{prop}The set of tripartite correlations which are LOSR-producible using sources $ {\boldsymbol{\psi}}_{\mathcal{Q}}^{\leq 2}$ is a strict superset of the tripartite correlations that are LOSR-producible using $\mathcal{S} \to {\boldsymbol{P}}_{\mathcal{Q}}^{\leq 2}$.
\end{prop}
\begin{proof}
The following proof makes use of the entanglement swapping, which is the paradigmatic advantage of sharing bipartite entanglement compared to sharing bipartite nonlocal correlations. The closest analog of entanglement swapping is nonlocal coupling~\cite{Skrzypczyk2009couplers,Short2010couplers}, but nonlocality coupling is not possible in a paradigm where local operations on boxes are limited to classical wirings. 
Consider a tripartite correlation with two settings for Alice and Charlie and three settings for Bob, ($x\in\{0,1\}$, $y\in\{0,1,2\}$, $z\in\{0,1\}$). Alice and Charlie always measure according to mutually unbiased bases. Bob, however, will ignore the singlet shared with Charlie for his first two settings, choosing instead measurements which lead to maximal violation of the CHSH inequality with Alice. For Bob's third setting, he performs a Bell-state measurement on the two singlets, coarse graining the outcome of that measurement to be $0$ if the postselected state on Alice and Charlie is the singlet, and $1$ otherwise. This ${\boldsymbol{\psi}}_{\mathcal{Q}}^{\leq 2}$ based strategy produces a correlation of the form:
\begin{equation}
P(abc|xyz) = 
\begin{cases}
 \frac{2+(-1)^{a \oplus b \oplus x y}\sqrt{2}}{16} & y\in\{0,1\} \\
 \frac{4 -2 (-1)^b + (-1)^{a\oplus b \oplus c \oplus x z}\sqrt{2}}{32} & y{=}2
\end{cases}
\end{equation}
We now argue that this correlation is not LOSR-producible using $\mathcal{S} \to {\boldsymbol{P}}_{\mathcal{Q}}^{\leq 2}$. 
From the maximal CHSH-inequality violation between Alice and Bob achieved when $y\in\{0,1\}$, we conclude that the measurements performed by Alice cannot depend in any way on the source that she shares with Charlie. On the other hand, we observe significant (maximal) CHSH-inequality violation between Alice and Charlie when we condition on $y{=}2$ and $b{=}0$. This Alice--Charlie nonlocality induced by postselection on Bob's measurement can \emph{only} be explained by entanglement swapping, since we have eliminated the possibility that $P(abc|xyz) $ utilizes any Alice--Charlie source.
\end{proof}

Let us conclude this section with a conjecture. Remark first that the set of correlations realizable using a network of $k$-way GPT sources (${\boldsymbol{\omega}}_{\mathcal{GPT}}^{\leq k}$) naturally includes the set of correlations realizable using a network of $k$-way nonsignalling boxes ($\mathcal{S} \to {\boldsymbol{P}}_{\mathcal{NS}}^{\leq k}$). We conjecture that here too, the inclusion is strict.
\begin{conj}
The set of tripartite correlation which are LOSR-producible using ${\boldsymbol{\omega}}_{\mathcal{GPT}}^{\leq 2}$ is a strict superset of the tripartite correlations that are LOSR-producible using $\mathcal{S} \to {\boldsymbol{P}}_{\mathcal{NS}}^{\leq 2}$.
\end{conj}

Our conjecture is based on the extremal-class \#4 of the set of extremal nonsignalling tripartite correlations, as enumerated in Ref.~\cite{Pironio2011tripartiteextremal}.
These correlations are known to be incompatible with a triangle network of type $\mathcal{S} \to {\boldsymbol{P}}_{\mathcal{NS}}^{\leq 2}$ per Ref.~\cite[Sec.~5-A]{FeatsFeaturesFailures}. Hence it is sufficient to prove they are LOSR-producible using ${\boldsymbol{\omega}}_{\mathcal{GPT}}^{\leq 2}$, that is, that they are 2-LOSR theory-agnostic. 
We found evidence that the linear constraint given by all the triangle inflations cannot rule out these correlations, suggesting they are at least \emph{weakly} 2-LOSR theory-agnostic correlations.

\subsection{Comparing and contrasting LOCC and LOSR producibility}\label{sec:SvetContrast}

If a correlation $P$ is in ${\boldsymbol{P}}_{\mathcal{R}}^{\leq k}$, \emph{i.e.}, LOCC-${\mathcal{R}}$ $k$-producible, then $P$ is also LOSR-producible using $\mathcal{S} \to {\boldsymbol{P}}_{\mathcal{R}}^{\leq k}$. The LOSR network which realizes the LOCC-relevant convex decomposition utilizes the shared randomness as a switch variable.
Accordingly, if $P$ is genuinely LOSR-$\mathcal{GPT}$ multipartite nonlocal, then $P$ is also certainly genuinely LOCC $k$-partite nonlocal relative to $\mathcal{R}\to {\mathcal{NS}}$. From the well-known containment of ${\boldsymbol{P}}_{\mathcal{Q}}\subset{\boldsymbol{P}}_{\mathcal{NS}}$, we further establish that if $P$ is genuinely LOSR-$\mathcal{GPT}$ multipartite nonlocal, then $P$ is also certainly genuinely LOCC $k$-partite nonlocal relative to ${\mathcal{Q}}$.

It is worth emphasizing that the implications about LOCC multipartite nonlocality from LOSR multipartite nonlocality run strictly \emph{one way}. That is, there are correlations which are genuinely LOCC-$\mathcal{NS}$ multipartite nonlocal which are \emph{not} genuinely LOSR-$\mathcal{NS}$ multipartite nonlocal. Perhaps the most famous example is the Svetlichny box; see Ref.~\cite[Fig.~5]{Barrett2005SvetFromPR}. Another simple example is the parallel composition of two distinct bipartite Tsirelson boxes, one for Alice--Bob and another for Bob--Charlie, as discussed in~\cite{PRL}. The resulting tripartite correlation (involving a 2-bit setting variable and 2-bit outcome variable for Bob) is genuinely LOCC-$\mathcal{NS}$ multipartite nonlocal~\cite{Tejada2020NetworkGMNL}, but, by construction, \emph{not} genuinely LOSR-$\mathcal{NS}$ multipartite nonlocal. 

It is also critical to recognize that \emph{Svetlichny} genuine multipartite nonlocality is \emph{not} implied by LOSR multipartite nonlocality. This is readily evident by noticing that inequality~\eqref{eq:ghz3-technical} is strongly violated by nonsignalling correlations generated via causal models wherein $a$ is allowed to functionally depend on $b$ and $y$ (hidden signalling from Bob to Alice). 
Consider the following fine-tuned (hidden signalling) local-hidden-variable model (LHVM): Let $\lambda$ indicate the value of the globally shared classical hidden random variable, such that $\lambda$ is uniformly distributed amongst the dichotomous values $\{+1,-1\}$. Let $c{=}\lambda$ always, \emph{i.e.}, for both $z\in\{0,1\}$; similarly, let $b{=}\lambda$ always, \emph{i.e.}, for all cases $y\in\{0,1,2\}$. Finally, let Alice's outcome depend on $y$ such that $a=b\times(-1)^{x y}$, an effect of which is that $a{=}b{=}c$ with unit probability for $y{=}2$.

For a further example, consider Box \#8 in the set of extremal nonsignalling tripartite correlations as enumerated in Ref.~\cite{Pironio2011tripartiteextremal}. Such correlations are known to be LOCC-producible using correlations in ${\boldsymbol{P}}_{\mathcal{S}\mathscr{ig}}^{\leq 2}$. Nevertheless, one can use nonfanout inflation to readily prove that such correlations are \emph{not} LOSR-producible within triangle networks using sources of type ${\boldsymbol{\omega}}_{\mathcal{GPT}}^{\leq 2}$. As such, Box \#8 is genuinely LOSR multipartite nonlocal yet not Svetlichny genuinely multipartite.

\section{Conclusion}

In this paper, we focused on correlations that cannot be obtained from arbitrary $(N-1)$-partite causal GPT resources and $N$-shared randomness, for any fixed $N$, which we called \emph{genuinely LOSR-multipartite-nonlocal correlations}. We proved that the (noisy) $\ket{{\rm GHZ}_N}$ states and the $\ket{{\rm W}}$ state can produce such correlations. This proves Theorem~\ref{thm:NatureNotNloc}, the main result of this paper: \emph{Nature is not merely ${N}$-partite, for any $N$}.
As this definition relies on an infinite hierarchy of nonlinear existence problems involving linear- and nonlinear- equality constraints of factorization, it is hard to manipulate in practice.
Using De Finetti's theorem, we obtained a nontrivial relaxation of the set of genuinely LOSR-multipartite-nonlocal correlations, which can be characterized by an infinite hierarchy of LP existence problems.
We illustrated its usefulness by improving the noise tolerance of our analysis of $\ket{{\rm GHZ}_3}$, making an experimental proof accessible to current technologies.
At last, we compared our introduced concept to already-existing definitions of genuine multipartite nonlocality.
We finish this paper with some comments and open questions.

Note that in our introduction, we argued that $N$-partite resource models of correlations should include classical shared randomness. This is motivated by the fact that, for instance, pre-established shared randomness can be stored in classical local memories.
We now remark that more general forms of randomness can \emph{a priori} be shared in the same way: For instance, pre-stored quantum states in quantum local memories can, in principle, also simulate a ``live" shared random quantum source. 
Certainly such unlimited quantum local memories are fundamentally more technologically demanding. Nevertheless, we also want to appeal to more foundations arguments for why the storage of many-partite GPT resources should be treated as costly; see, for example, the resource theory of quantum memory developed in Ref.~\cite{QuantumMemoryRT}. The trade-off between nonclassical-memory capacity and the resource value of genuinely LOSR multipartite theory-agnostic correlations is a topic we highlight for future research.

The connections between our own definition of causal GPT in networks, which is based on the concepts of causality and device replication, and the standard GPT framework~\cite{Skrzypczyk2009couplers,Short2010couplers,Barrett2007GPT,Chiribella2011Reconstruction,Janotta2012Boxword} are also left for future work~\cite{Pironio2021InPreparation}.

Motivated by a desire to concretely formulate computational algorithms, we relaxed the set of LOSR theory-agnostic Correlations into the set of \emph{weakly} LOSR theory-agnostic correlations. 
The question of the differences between these two sets remains open.
It might be that a refined version of the argument based on De Finetti's theorem could prove that the two coincide.
It might also be that our relaxation is strict and that there exists a correlation that is in the relaxed set without being in the original one.
We expect that such approaches will allow to find better noise-tolerant results for practical experimental demonstrations that Nature is not merely genuinely LOSR $N$-partite nonlocal for low $N$~\cite{GHZExperiment6Photons,Chao2019,Proietti2019}, and will allow to extend our analytical proofs to more quantum states, such as the generalization of the tripartite $\ket{\W}$ state.

\emph{Acknowledgements.---}
We thank Claude Crépeau, Nicolas Gisin, Miguel Navascués, and Stefano Pironio for valuable discussions.
This research was supported by the Swiss National Science Foundation (SNF), the Fonds de recherche du Qu\'ebec -- Nature et technologies (FRQNT), and Perimeter Institute for Theoretical Physics. Research at Perimeter Institute is supported in part by the Government of Canada through the Department of Innovation, Science and Economic Development Canada and by the Province of Ontario through the Ministry of Colleges and Universities. M.-O.R.~is supported by the Swiss National Fund Early Mobility Grant P2GEP2\_191444 and acknowledges the Government of Spain (FIS2020-TRANQI and Severo Ochoa CEX2019-000910-S), Fundació Cellex, Fundació Mir-Puig, Generalitat de Catalunya (CERCA, AGAUR SGR 1381) and the ERC AdG CERQUTE.

\bigskip
\nocite{apsrev42Control}
\setlength{\bibsep}{1pt plus 1pt minus 1pt}
\bibliographystyle{apsrev4-2-wolfe}
\bibliography{GPTandOtherConstraints}

\end{document}